\documentclass[10pt, conference]{IEEEtran} % try this
\usepackage{url}% my addition
\usepackage{makecell}% my addition
\usepackage{amsfonts}
\usepackage{amssymb}
\usepackage{amsmath}
\usepackage{graphicx, colordvi, psfrag}% my addition
\usepackage{calc,pstricks, pgf, xcolor}% my addition
\usepackage{epsfig, cite}% my addition
\usepackage{bm} % my addition
\usepackage{bbm} % my addition
\usepackage{enumerate}
\newcommand{\beq}[1]{\begin{equation}\label{#1}}
\newcommand{\eeq}{\end{equation}}
\newcommand{\beqn}[1]{\begin{eqnarray}\label{#1}}
\newcommand{\eeqn}{\end{eqnarray}}

\newtheorem{thmbody}{Theorem}
\newenvironment{thm}{
\begin{thmbody}
	}{
	\end{thmbody} % \end{singlespace}
	}
\newtheorem{dfnbody}{Definition}

\newtheorem{corbody}{Corollary}

\newtheorem{lemmabody}{Lemma}
\newenvironment{lemma}{
\begin{lemmabody}
	}{
	\end{lemmabody} %  \end{singlespace}
	}
\newtheorem{propbody}{Proposition}

\newenvironment{proof}{
	{\it Proof:}
	}{
 $\Box$
	}

\usepackage{dsfont}
\begin{document}
\title{Proof of Convergence for\\ Correct-Decoding Exponent Computation}

\author{\IEEEauthorblockN{Sergey Tridenski}
\IEEEauthorblockA{Faculty of Engineering\\Bar-Ilan University, Israel\\
Email: tridens@biu.ac.il}
\and
\IEEEauthorblockN{Anelia Somekh-Baruch}
\IEEEauthorblockA{Faculty of Engineering\\Bar-Ilan University, Israel\\
Email: somekha@biu.ac.il}
\and
\IEEEauthorblockN{Ram Zamir}
\IEEEauthorblockA{EE - Systems Department\\Tel-Aviv University, Israel\\
Email: zamir@eng.tau.ac.il}}

%\author{
%Sergey~Tridenski, Anelia~Somekh-Baruch, and %~\IEEEmembership{Fellow},~\IEEEmembership{IEEE},
%Ram~Zamir}

%\thanks{
%The material in this paper was partially presented in ISIT2017 \cite{TridenskiZamir17} and ISIT2018 \cite{TridenskiZamir18}.}
%\thanks{This work of S. Tridenski and R. Zamir was partially supported by the Israel Science Foundation (ISF), grant \# 676/15,
%and by the US-Israel Binational and US-National Science Foundations (BSF-NSF), grant \# 2018690.}
%}

\maketitle
\begin{abstract}
%We prove convergence of iterative minimization at fixed rate and at fixed slope to the reliability function above capacity
%for a discrete memoryless channel with finite input and output alphabets.
For a discrete memoryless channel with finite input and output alphabets,
we prove convergence of a parametric family of iterative computations of the %reliability function above capacity
optimal correct-decoding exponent. The exponent, as a function of communication rate,
is computed for a fixed rate and for a fixed slope.
\end{abstract}

%\begin{IEEEkeywords}
%Correct-decoding exponent, Arimoto algorithm, Blahut algorithm.
%%,\newline maximum mutual information, erasure decoder.
%\end{IEEEkeywords}

%\markboth
%{To be Submitted to the IEEE Trans. on Information Theory}
%{Submitted to the IEEE Trans. on Information Theory, Revised April 2019}
%{Tridenski and Zamir: Channel input adaptation via natural type selection} % TODO: use the final title


%\IEEEPARstart{T}{he}

%\section{Introduction}

\bigskip

\section{Introduction}\label{Intro}

\bigskip

Consider a standard information theoretic setting of transmission through a discrete memoryless channel (DMC), with finite input and output alphabets, using block codes.
%Let $P(y \, | \,x)$ denote its transition probabilities.
For communication rates above capacity, the average probability of correct decoding in a block code tends to zero exponentially %with increase of block length.
fast as a function of the block length.
In the limit of a large block length, the lowest possible exponent %in
corresponding to
the probability of correct decoding, also called the reliability function above capacity, for all\footnote{The expression gives zero for the rates $R \leq \max_{\,Q}I(Q, P)$.} rates $R \geq 0$ is given by \cite{DueckKorner79}
\begin{equation} \label{eqIntroduction}
{E\mathstrut}_{\!c}(R) =
\min_{\substack{\\Q(x), \\ W(y\,|\,x)}}
\left\{
D( W \, \| \, P \, | \, Q) \, + \,
\big|R - I(Q, W)\big|^{+}
\right\},
%\label{eqDCcordecexpOriginal}
\end{equation}
where $P$ denotes the channel's transition probability $P(y\,|\,x)$, %in a DMC,
$D( W \, \| \, P \, | \, Q)$ is the Kullback-Leibler divergence between the conditional distributions $W$ and $P$, averaged over $Q$,
%and
$I(Q, W)$ is the %Shannon
mutual information of a pair of random variables with a joint distribution $Q(x)W(y\,|\,x)$,
%Also
and $|t|^{+} = \max\,\{0, t\}$.

%Sometimes
For %some
certain
applications,
it is important to be able to know the actual value of ${E\mathstrut}_{\!c}(R)$ when it is positive.
For example, in
applications of secrecy, it might be interesting to know the correct-decoding exponent of %at
an
eavesdropper.
Several algorithms have been proposed for computation of ${E\mathstrut}_{\!c}(R)$.

In the %computation
algorithm by Arimoto \cite{Arimoto76}
the computation of ${E\mathstrut}_{\!c}(R)$ is facilitated by an alternative expression for it \cite{Arimoto73}, \cite{DueckKorner79}, %\cite{OohamaJitsumatsu15},
\cite{OohamaJitsumatsu19}:
\begin{equation} \label{eqGallagerForm}
{E\mathstrut}_{\!c}(R) =
\sup_{\substack{\\0 \, \leq \, \rho \, < \, 1}}
\min_{\substack{\\Q}}
\big\{
{E\mathstrut}_{0}(-\rho, Q) \, + \, \rho R
\big\},
%\label{eqDCcordecexpOriginal}
\end{equation}
where ${E\mathstrut}_{0}(-\rho, Q)$ is the Gallager exponent function \cite[Eq.~5.6.14]{Gallager}.
In \cite{Arimoto76},
$\min_{\,Q}{E\mathstrut}_{0}(-\rho, Q)$ is computed for a fixed slope parameter $\rho$.
The computation is performed iteratively as alternating minimization,
based on the property that $\min_{\,Q}{E\mathstrut}_{0}(-\rho, Q)$ can be written as a double
minimum:
\begin{equation} \label{eqArimoto}
\min_{\substack{\\Q}} \; \min_{\substack{\\V}} \; \bigg\{-\log \sum_{x,\,y} Q^{1-\rho}(x)V^{\rho}(x\,|\,y)P(y\,|\,x)\bigg\},
\end{equation}
where the inner minimum is in fact equal to ${E\mathstrut}_{0}(-\rho, Q)$.
In \cite{OohamaJitsumatsu19}, \cite{OohamaJitsumatsu15} a different alternating-minimization algorithm is introduced,
based on the property, that $\min_{\,Q}{E\mathstrut}_{0}(-\rho, Q)$ can be written %alternatively
as another double
minimum over distributions:
\begin{equation} \label{eqOohama}
\min_{\substack{\\T, \,V}} \,\min_{\substack{\\T_{1},\,V_{1}}} \bigg\{-\sum_{x,\,y}T(y)V(x\,|\,y)\log\frac{V^{\rho}_{1}(x\,|\,y)P(y\,|\,x)}{U^{\rho-1}_{1}(x)T(y)V(x\,|\,y)}\bigg\},
\end{equation}
where ${U\mathstrut}_{\!1}(x) = \sum_{y}{T\mathstrut}_{\!1}(y){V\mathstrut}_{\!1}(x\,|\,y)$.
As with (\ref{eqArimoto}), the computation of ${E\mathstrut}_{\!c}(R)$ with (\ref{eqOohama}) is also performed for a fixed $\rho$.

Sometimes, however, it is suitable or desirable to compute ${E\mathstrut}_{\!c}(R)$ directly for a given rate $R$.
For example, when ${E\mathstrut}_{\!c}(R) = 0$, and we would like to find such a distribution $Q$, for which the minimum (\ref{eqIntroduction})
is zero, as a by-product of the computation. Such distribution $Q$ has a practical meaning of a channel input distribution achieving reliable communication.
In \cite{TridenskiZamir18Trans}, an iterative minimization procedure for computation of ${E\mathstrut}_{\!c}(R)$ at fixed $R$ is proposed,
using the property that ${E\mathstrut}_{\!c}(R)$ can be written as a double minimum \cite{TridenskiZamir17}:
\begin{equation} \label{eqProperty}
\min_{\substack{\\Q(x)}}
\;
\min_{\substack{\\T(y), \\ V(x\,|\,y)
}}
\;
\left\{
D( TV \, \| \, QP) \, + \,
\big|R - D(V\,\|\,Q\,|\,T)\big|^{+}
\right\},
\end{equation}
where the inner $\min$ equals $\,\sup_{\,0 \, \leq \, \rho \, < \, 1}
\big\{
{E\mathstrut}_{0}(-\rho, Q) \, + \, \rho R
\big\}$.
In \cite{TridenskiZamir18Trans},
the inner minimum of (\ref{eqProperty}) is computed stochastically
by virtue of a correct-decoding {\em event} itself,
yielding the minimizing solution ${T\mathstrut}^{*}{V\mathstrut}^{*}$.
%the minimizing solution ${T\mathstrut}^{*}{V\mathstrut}^{*}$ of
%the inner minimum of (\ref{eqProperty}) is approximated stochastically
%by virtue of a correct-decoding {\em event} itself. %, yielding the minimizing solution ${T\mathstrut}^{*}{V\mathstrut}^{*}$.
The computation is then repeated iteratively,
by assigning $Q(x) = \sum_{y}{T\mathstrut}^{*}(y){V\mathstrut}^{*}(x\,|\,y)$.
It is shown in %\cite{TridenskiZamir18Trans}
\cite[Theorem~1]{TridenskiZamir18Trans}, that
the iterative procedure using the inner minimum of (\ref{eqProperty}) leads to convergence of this minimum
%this iterative procedure leads to convergence of the inner minimum of (\ref{eqProperty})
to the double minimum (\ref{eqProperty}), which is evaluated at least over some {\em subset} of the support of the initial distribution ${Q\mathstrut}_{0}$.
In addition, a sufficient condition on ${Q\mathstrut}_{0}$ is provided, which guarantees convergence of the inner minimum in (\ref{eqProperty})
to zero.
This condition on ${Q\mathstrut}_{0}$ in %\cite{TridenskiZamir18Trans}
\cite[Lemma~6]{TridenskiZamir18Trans}
is rather %weak and
limiting, and is hard to verify.
%evaluate.

In the current work, we %do away with the sufficient condition from \cite{TridenskiZamir18Trans}
improve the result of \cite{TridenskiZamir18Trans}.
We modify the method of Csisz\'ar and Tusn\'ady \cite{CsiszarTusnady84} to
%and
prove that %, in fact,
the iterative minimization procedure of \cite{TridenskiZamir18Trans} converges to the global minimum (\ref{eqProperty})
over the {\em support} of the initial distribution ${Q\mathstrut}_{0}$ {\em itself}, for any $R$ (i.e., not only
if the global minimum is zero),
%to ${E\mathstrut}_{\!c}(R) = 0$
%zero),
and without any additional condition.
In particular, use of a strictly positive %initial distribution
${Q\mathstrut}_{0}$ guarantees convergence to ${E\mathstrut}_{\!c}(R)$.

By a similar method, we also show convergence of the fixed-slope counterpart of the minimization (\ref{eqProperty}), which is an alternating minimization at fixed $\rho$,
based on the double minimum \cite{TridenskiZamir18Extended}
\begin{equation} \label{eqForCompare}
\min_{\substack{\\Q}} \;\min_{\substack{\\T,\,V}}\; \bigg\{-\sum_{x,\,y}T(y)V(x\,|\,y)\log\frac{Q^{1-\rho}(x)P(y\,|\,x)}{T(y)V^{1-\rho}(x\,|\,y)}\bigg\},
\end{equation}
where the inner minimum is in fact equal to ${E\mathstrut}_{0}(-\rho, Q)$.
%Unlike the minimization with (\ref{eqArimoto}), as $\rho \rightarrow 0$, the minimization with (\ref{eqForCompare})
%does not turn into an algorithm for computation of the channel capacity.

Furthermore, in the current paper we extend the analysis, presented in the shorter version of the paper \cite{TridenskiSomekhBaruchZamir20}.
Here we slightly generalize the expression (\ref{eqProperty}).
Using this generalization, we prove convergence of a parametric family of iterative computations,
of which the computation according to (\ref{eqProperty}) from \cite{TridenskiZamir18Trans},
as well as the computations according to (\ref{eqForCompare}), \cite{TridenskiZamir18Extended}, and according to (\ref{eqOohama}), \cite{OohamaJitsumatsu19},
%from \cite{TridenskiZamir18Extended},
%and according to (\ref{eqOohama}) from \cite{OohamaJitsumatsu19}, \cite{OohamaJitsumatsu15},
become special cases.

As in the shorter version of the paper \cite{TridenskiSomekhBaruchZamir20}, besides the variable $R$, we take into account also a possible channel-input constraint, denoted by $\alpha$.
In Section~\ref{CorDecExp} we examine the expression for the correct-decoding exponent.
In Section~\ref{FixedRate} we prove convergence of the iterative minimization for fixed $(R, \alpha)$.
In Section~\ref{FixedSlope} we prove convergence of the iterative minimization for fixed gradient w.r.t. $(R, \alpha)$.
In Sections~\ref{FixedMixed} and ~\ref{FixedMixed2} we prove convergence of mixed scenarios:
for fixed $\alpha$ and slope $\rho$ in the direction of $R$, and vice versa.

%\begin{equation} \label{eqBlahut}
%\min_{\substack{\\Q}} \; \min_{\substack{\\V}}\; \bigg\{-\sum_{x,\,y}P(y)V(x\,|\,y)\log\frac{Q(x)e^{-\rho d(y, \,x)}}{V(x\,|\,y)}\bigg\}
%\end{equation}

%\bigskip

\section{Correct-decoding exponent}\label{CorDecExp}
%\subsection{Correct-decoding exponent}\label{CorDecExp}
%\bigskip
%\subsection{Implicit expressions}\label{CorDecExpA}

%\bigskip

Let $P(y\,|\,x)$ denote transition probabilities in a DMC from $x \in {\cal X}$ to $y \in {\cal Y}$,
where ${\cal X}$ and ${\cal Y}$ are finite channel input and output alphabets, respectively.
Suppose also that the channel input $x$ with an additive cost function $f: {\cal X} \rightarrow \mathbb{R}$ satisfies on average an input constraint $\alpha \in \mathbb{R}$, chosen large enough, such that  $\alpha \geq \min_{\,x} f(x)$.
The maximum-likelihood %(ML)
correct-decoding exponent (\cite{DueckKorner79}, \cite{Oohama17}) of this channel, as a function of the rate $R\geq 0$ and the input constraint $\alpha$, is given by
\begin{align}
& {E\mathstrut}_{\!c}(R, \alpha) \; =
\label{eqDCcordecexp} \\
&\min_{\substack{\\Q(x):\\\mathbb{E}_{Q}[f(X)]\;\leq \; \alpha}}
\;
\min_{\substack{\\W(y\,|\,x)
}}
\;
\left\{
D( W \, \| \, P \, | \, Q) \, + \,
\big|R - I(Q, W)\big|^{+}
\right\},
\nonumber
\end{align}
where %$I(Q, W)$ is the Shannon mutual information of a pair of random variables with the joint distribution $Q(x)W(y \, | \, x)$,
%$D( QW \, \| \, QP)$ denotes the Kullback-Leibler
%divergence between the joint distributions $Q(x)W(y \, | \, x)$ and $Q(x)P(y \, | \, x)$,
%denoted as $QW$ and $QP$, respectively,
%and
$\mathbb{E}_{Q}[f(X)]$ denotes the expectation of $f(x)$ w.r.t. the distribution $Q(x)$ over ${\cal X}$.
%Also $|t|^{+} = \max\,\{0, t\}$. %$|\,a - b\,|^{+} = \max\,\{a - b, \,0\}$.

Let $Q(x)W(y\,|\,x) \equiv T(y)V(x\,|\,y)$, or $QW$, denote a distribution over ${\cal X}\times {\cal Y}$,
and let $\widetilde{Q}\,\widetilde{\!W}$ be another such distribution.
We can think of $4$ different divergences {\em from} $\widetilde{Q}\,\widetilde{\!W}$ {\em to} $QW$:
$D(Q\,\|\, \widetilde{Q})$, $D(W \, \| \, \,\widetilde{\!W} \, | \, Q)$, $D(T\,\|\,\widetilde{T})$,
and $D(V \,\| \, \widetilde{V} \, | \, T)$. Using $4$ non-negative parameters $t_{i}\geq 0$, $i = 1, 2, 3, 4$,
we define a non-negative linear combination of these divergences:
\begin{align}
{D\mathstrut}^{\bf t}(QW, \,\widetilde{Q}\,\widetilde{\!W}) \; \triangleq \;\; & t_{1}D(Q\,\|\, \widetilde{Q}) +
t_{2}D(W \, \| \, \,\widetilde{\!W} \, | \, Q) \, +
\nonumber \\
&
t_{3}D(T\,\|\, \widetilde{T}) +
t_{4}D(V \, \| \, \widetilde{V} \, | \, T),
\label{eqCombination}
\end{align}
where ${\bf t} \triangleq (t_{1}, t_{2}, t_{3}, t_{4})$ is an index.
With the help of ${D\mathstrut}^{\bf t}(QW, \widetilde{Q}\,\widetilde{\!W})$,
the expression (\ref{eqDCcordecexp}) can be rewritten as follows:
\begin{align}
&
\;\;\;\;\;\;\;\;\;\;\;\,\,\,\,
\min_{\substack{\\Q, \, W:\\\mathbb{E}_{Q}[f(X)]\;\leq \; \alpha}}
\;
\left\{
D( W \, \| \, P \, | \, Q) \, + \,
\big|R - I(Q, W)\big|^{+}
\right\}
\nonumber \\
& = \;
\;\;\;\;\;\;\;\,\,
\min_{\substack{\\Q, \, W:\\
\mathbb{E}_{Q}[f(X)]\;\leq \; \alpha
}}
\!
\max
\Big\{D( W \, \| \, P \, | \, Q),
\nonumber \\
&
\;\;\;\;\;\;\;\;\;\;\;\;\;\;\;\;\;\;\;\;\;\;\;\;\;\;\;\;
\;\;\;\;\;\;\;\;\;\;\;\,\,\,\,
D( W \, \| \, P \, | \, Q) + R - I(Q, W)
\Big\}
\nonumber
\end{align}
\begin{align}
& = \;
\min_{\substack{\\\widetilde{Q}, \, \,\widetilde{\!W}}}
\;
\min_{\substack{\\Q, \, W:\\
\mathbb{E}_{Q}[f(X)]\;\leq \; \alpha
}}
\!
\max
\Big\{D( W \, \| \, P \, | \, Q) %\,
 + {D\mathstrut}^{\bf t}(QW, \,\widetilde{Q}\,\widetilde{\!W}),
\nonumber \\
&
\;\;\;\;\;\;\;\;\;\;\;\;\;\,\,\,\,
%{D\mathstrut}_{\bf t}(UW, \,\widetilde{U}\,\widetilde{\!W}), \;\;
\;\;\;\;\;\;\;\;\;\;\;\;\;\;\;\;\;\;\;\;\;\;\;\;\;\;
D( W \, \| \, P \, | \, Q) + R - I(Q, W)
\Big\},
\label{eqEquality}
\end{align}
where the first equality holds because $|a|^{+} = \max\,\{0, a\}$,
and the second equality follows since
$\min_{\widetilde{Q}\,\widetilde{\!W}} {D\mathstrut}^{\bf t}(QW, \widetilde{Q}\,\widetilde{\!W}) = 0$
and the minima can be interchanged.
In \cite{TridenskiZamir18Trans} a special case (${\bf t} = (1, 0, 0, 0)$) of the inner minimum of (\ref{eqEquality}) was used
as a basis of
an iterative
procedure to
find minimizing solutions of (\ref{eqDCcordecexp}).
In what follows, we modify the method of Csisz\'ar and Tusn\'ady \cite{CsiszarTusnady84}
to show convergence of that minimization procedure.
The method allows us to prove convergence in a slightly more general setting (\ref{eqEquality}), (\ref{eqCombination}), with arbitrary non-negative parameters $(t_{1}, t_{2}, t_{3}, t_{4})$.

%\bigskip

\section{Convergence of the iterative minimization for fixed $(R, \alpha)$}\label{FixedRate}

%\bigskip

Let us define a short notation for the maximum in (\ref{eqEquality}):
%, which is also the objective function of (\ref{eqBasis}):
\begin{align}
{F\mathstrut}_{1}^{\bf t}(QW, \,\widetilde{Q}\,\widetilde{\!W}) \, & \triangleq \, D( W \, \| \, P \, | \, Q) %\,
 + {D\mathstrut}^{\bf t}\!(QW, \,\widetilde{Q}\,\widetilde{\!W}),
\label{eqF1} \\
%\end{align}
%\begin{align}
{F\mathstrut}_{2}(QW, R) \, & \triangleq \, D( W \, \| \, P \, | \, Q) - I(Q, W) + R,
\label{eqF2} \\
%\end{align}
%\begin{align}
{F\mathstrut}^{\bf t}(QW, \,\widetilde{Q}\,\widetilde{\!W}, \, R) \, & \triangleq \, \max \Big\{{F\mathstrut}_{1}^{\bf t}(QW, \,\widetilde{Q}\,\widetilde{\!W}),  \,{F\mathstrut}_{2}(QW, R)\Big\}.
\label{eqF}
\end{align}
Define %also
notation for the inner minimum in (\ref{eqEquality}):
\begin{equation} \label{eqDef}
{E\mathstrut}_{\!c}^{\bf t}(\widetilde{Q}\,\widetilde{\!W}, R, \alpha) \; \triangleq \;
\min_{\substack{\\Q, \, W:\\
\mathbb{E}_{Q}[f(X)]\;\leq \; \alpha
}}
{F\mathstrut}^{\bf t}(QW, \,\widetilde{Q}\,\widetilde{\!W}, \, R)
\end{equation}
The iterative minimization procedure from \cite{TridenskiZamir18Trans}, consisting of two steps in each iteration\footnote{Note that (\ref{eqMinProcedure}) %this
is not just an alternating minimization procedure w.r.t. ${F\mathstrut}^{\bf t}(QW, \widetilde{Q}\,\widetilde{\!W}, R)$, or not the only one possible,
in a sense that other choices of ${\widetilde{Q}\mathstrut}_{\ell\,+\,1}{\,\widetilde{\!W}\mathstrut}_{\!\ell\,+\,1}$ may also minimize ${F\mathstrut}^{\bf t}({Q\mathstrut}_{\ell}{W\mathstrut}_{\!\ell}, \,\cdot\;, \, R)$.
%For example, in the absence of the channel input constraint, for any $Q$ it already holds that $F({U\mathstrut}_{\!\ell}\,{W\mathstrut}_{\!\ell}, \,Q, \, %R) \geq F({U\mathstrut}_{\!\ell}\,{W\mathstrut}_{\!\ell}, \,{Q\mathstrut}_{\ell}, \, R)$,
%and, in particular, any $Q$,
%such that $D({U\mathstrut}_{\!\ell} \,\|\, Q) \leq D({U\mathstrut}_{\!\ell} \,\|\, {Q\mathstrut}_{\ell})$,
%will minimize $F({U\mathstrut}_{\!\ell}\,{W\mathstrut}_{\!\ell}, \,Q, \, R)$.
},
%$F({U\mathstrut}_{\!\ell}\,{W\mathstrut}_{\!\ell}, \,{Q\mathstrut}_{\ell}) = F({U\mathstrut}_{\!\ell}\,{W\mathstrut}_{\!\ell}, %\,{Q\mathstrut}_{\ell\,+\,1})$.},
in a more general form is given by
\begin{align}
&
\begin{array}{l}
\displaystyle
\;\;\;\;\;\;\;\;\;
{Q\mathstrut}_{\ell}{W\mathstrut}_{\!\ell} \;\; \in \;
\underset{\substack{\\Q, \, W:\\
\mathbb{E}_{Q}[f(X)]\;\leq \; \alpha
}}{\arg\min}
%\arg \min_{\substack{\\U(x), \, W(y\,|\,x):\\
%\mathbb{E}_{U}[f(X)]\;\leq \; \alpha
%}}
{F\mathstrut}^{\bf t}(QW, \, {\widetilde{Q}\mathstrut}_{\ell}{\,\widetilde{\!W}\mathstrut}_{\!\ell}, \, R), \\ %\;\;\; \ell \, = \, 0, 1, 2, ...
%\label{eqStep1} \\
\,{\widetilde{Q}\mathstrut}_{\ell\, + \, 1}{\,\widetilde{\!W}\mathstrut}_{\!\ell\, + \, 1} \; = \;
\;\;\;\;
{Q\mathstrut}_{\ell}{W\mathstrut}_{\!\ell},
%\;\;\;\;\;\;\;\;\;\;\;\;\;\;\;\;\;\;\;\;\;\;\;\;\;\;\;\;\;
%\ell \, = \, 0, 1, 2, ... \,.
\end{array}
\label{eqMinProcedure} \\
&
\;\;\;\;\;\;\;\;\;\;\;\;\;\;\;\;\;\;\;
\;\;\;\;\;\;\;\;\;\;\;\;\;\;\;\;\;\;\;
\;\;\;\;\;\;\;\;\;\;\;\;\;\;\;\;\;\;\;
\ell \, = \, 0, 1, 2, ... \,.
\nonumber
\end{align}
We assume that ${\widetilde{Q}\mathstrut}_{0}{\,\widetilde{\!W}\mathstrut}_{\!0}$ in (\ref{eqMinProcedure}) is chosen such that
the set $\big\{QW:  \sum_{x}Q(x)f(x) \leq \alpha, \; {F\mathstrut}_{1}^{\bf t}(QW, \,{\widetilde{Q}\mathstrut}_{0}{\,\widetilde{\!W}\mathstrut}_{\!0}) < +\infty \big\}$
is non-empty, which guarantees
${F\mathstrut}^{\bf t}({Q\mathstrut}_{0}{W\mathstrut}_{\!0}, \,{\widetilde{Q}\mathstrut}_{0}{\,\widetilde{\!W}\mathstrut}_{\!0}, \, R) = {E\mathstrut}_{\!c}^{\bf t}({\widetilde{Q}\mathstrut}_{0}{\,\widetilde{\!W}\mathstrut}_{\!0}, R, \alpha) < +\infty$.
By (\ref{eqF1}) it is clear that (\ref{eqMinProcedure}) produces a monotonically non-increasing sequence ${E\mathstrut}_{\!c}^{\bf t}({\widetilde{Q}\mathstrut}_{\ell}{\,\widetilde{\!W}\mathstrut}_{\!\ell}, R, \alpha)$, $\ell = 0, 1, 2, ... \,$.
Our main result is given by the following theorem, which is an improvement on \cite[Theorem~1]{TridenskiZamir18Trans} and \cite[Lemma~6]{TridenskiZamir18Trans}:

\bigskip

\begin{thm}%[]
\label{thmConvergence}%\newline
{\em Let ${\big\{{Q\mathstrut}_{\ell}{W\mathstrut}_{\!\ell}\big\}\mathstrut}_{\ell \, = \, 0}^{+\infty}$ be a sequence of iterative solutions produced by (\ref{eqMinProcedure}). %\newline
%If $\,{E\mathstrut}_{\!c}({Q\mathstrut}_{0}, R, \alpha) < +\infty$, %\newline
Then}
\begin{equation} \label{eqConvergence}
%F({U\mathstrut}_{\!\ell}\,{W\mathstrut}_{\!\ell}, \,{Q\mathstrut}_{\ell}, \, R)
{E\mathstrut}_{\!c}^{\bf t}({\widetilde{Q}\mathstrut}_{\ell}{\,\widetilde{\!W}\mathstrut}_{\!\ell}, R, \alpha)
\, \overset{\ell \, \rightarrow\,\infty}{\searrow} \,
\min_{\substack{\\\widetilde{Q}, \, \,\widetilde{\!W}:\\ {D\mathstrut}^{\bf t}(\widetilde{Q}\,\widetilde{\!W}\!, \,{\widetilde{Q}\mathstrut}_{0}{\,\widetilde{\!W}\mathstrut}_{\!0})
\,<\, \infty}}
{E\mathstrut}_{\!c}^{\bf t}(\widetilde{Q}\,\widetilde{\!W}, R, \alpha),
%\;\;
%\min_{\substack{\\U(x), \, W(y\,|\,x):\\
%\mathbb{E}\,[f(X)]\;\leq \; \alpha
%}}
%\;
%F(UW, Q).
\end{equation}
%\begin{equation} \label{eqBoundEll}
%F({U\mathstrut}_{\!\ell}\,{W\mathstrut}_{\!\ell}, \,{Q\mathstrut}_{\ell})
%\;\; \leq \;\; F(\hat{U}\hat{W}, \hat{U}) \, + \, \big|D( \hat{U}\hat{W} \, \| \, {Q\mathstrut}_{\ell} P) - D( \hat{U}\hat{W} \, \| \, %{Q\mathstrut}_{\ell\,+\,1} P)\big|^{+}, \;\;\;\;\;\; \ell \, = \, 0, 1, 2, ... \,.
%\end{equation}
{\em where ${E\mathstrut}_{\!c}^{\bf t}(\widetilde{Q}\,\widetilde{\!W}, R, \alpha)$ is defined in (\ref{eqDef}) and ${D\mathstrut}^{\bf t}(\cdot \,, \cdot)$  in (\ref{eqCombination}).}
\end{thm}

\bigskip

%Throughout the paper, we also use notation $\text{supp}(QW) \triangleq \{(x, y)\in {\cal X}\times {\cal Y}: Q(x)W(y\, | \, x) > 0\}$.
Suppose ${Q\mathstrut}^{*}{W\mathstrut}^{*}$ is a minimizing solution of (\ref{eqDCcordecexp}).
If the initial distribution $\widetilde{Q}_{0}\,\widetilde{\!W}_{\!0}$ in the iterations (\ref{eqMinProcedure}) is chosen such that ${D\mathstrut}^{\bf t}({Q\mathstrut}^{*}{W\mathstrut}^{*}\!, \, \widetilde{Q}_{0}\,\widetilde{\!W}_{\!0}) < +\infty$
(for example, if $\text{support}(\widetilde{Q}_{0}\,\widetilde{\!W}_{\!0}) = {\cal X}\times {\cal Y}$), then
by (\ref{eqEquality})
%for any vector of non-negative parameters ${\bf t}$,
the RHS of (\ref{eqConvergence})
gives (\ref{eqDCcordecexp}).
The choice of ${\bf t} = (1, 0, 0, 0)$ in (\ref{eqCombination}) corresponds to the iterative minimization in \cite{TridenskiZamir18Trans}.
In order to prove Theorem~\ref{thmConvergence}, we use a lemma, which is similar to ``the five points property'' from \cite{CsiszarTusnady84}.

\bigskip

\begin{lemma} %[]
\label{lem5PP}%\newline
{\em Let $\hat{Q}\hat{W}$ be such, that $\,\sum_{x}\hat{Q}(x)f(x) \leq \alpha$
and ${F\mathstrut}_{1}^{\bf t}(\hat{Q}\hat{W}, \,{\widetilde{Q}\mathstrut}_{0}{\,\widetilde{\!W}\mathstrut}_{\!0}) < +\infty$.
%If ${F\mathstrut}_{1}({U\mathstrut}_{\!0}{W\mathstrut}_{\!0}, \,{Q\mathstrut}_{0}) > {F\mathstrut}_{2}({U\mathstrut}_{\!0}{W\mathstrut}_{\!0}, \,R)$, then %$\text{\em supp}(\hat{U}) \subseteq \text{\em supp}({Q\mathstrut}_{1})$ and
Then}
\begin{align}
& {F\mathstrut}^{\bf t}({Q\mathstrut}_{0}{W\mathstrut}_{\!0}, \,{\widetilde{Q}\mathstrut}_{0}{\,\widetilde{\!W}\mathstrut}_{\!0}, \, R)
\;\; \leq \;\;
{F\mathstrut}^{\bf t}(\hat{Q}\hat{W}, \hat{Q}\hat{W}, R)
\nonumber \\
&
\;\;\;\;\;\;\;\;\;\;\;\;\;\;\,
+ \big|{F\mathstrut}_{1}^{\bf t}( \hat{Q}\hat{W} , \, {\widetilde{Q}\mathstrut}_{0}{\,\widetilde{\!W}\mathstrut}_{\!0})  -   {F\mathstrut}_{1}^{\bf t}( \hat{Q}\hat{W} , \, {\widetilde{Q}\mathstrut}_{1}{\,\widetilde{\!W}\mathstrut}_{\!1})\big|^{+}.
\label{eqBound1}
\end{align}
\end{lemma}

\bigskip

\begin{proof}
Let us define a set of distributions $QW$:
\begin{align}
&
{\cal S} \triangleq \bigg\{QW: \; \sum_{x}Q(x)f(x) \leq \alpha,
\;
{F\mathstrut}_{1}^{\bf t}(QW, \,{\widetilde{Q}\mathstrut}_{0}{\,\widetilde{\!W}\mathstrut}_{\!0}) < +\infty\bigg\}.
\nonumber
%&
%\;\;\;\;\;\;\;\;\;\;\;\;\;\;\;\;\;\;\;\;\;\;\;\;\;
%\text{supp}(QW) \subseteq \text{supp}(QP)
%\cap \text{supp}({\widetilde{Q}\mathstrut}_{0}{\,\widetilde{\!W}\mathstrut}_{\!0})
%\bigg\}.
%\nonumber
\end{align}
Observe that ${\cal S}$ is a closed convex set. Since $\hat{Q}\hat{W} \in {\cal S}$,
then ${\cal S}$ is non-empty and by (\ref{eqMinProcedure}) we have also that ${Q\mathstrut}_{0}{W\mathstrut}_{\!0} \in {\cal S}$.
Observe further that %both
the two
terms in the maximization of (\ref{eqF}),
${F\mathstrut}_{1}^{\bf t}(QW, \,{\widetilde{Q}\mathstrut}_{0}{\,\widetilde{\!W}\mathstrut}_{\!0})$
and ${F\mathstrut}_{2}(QW, R)$, as functions of $QW$, are convex ($\cup$) and continuous in ${\cal S}$.

Consider the case ${F\mathstrut}_{1}^{\bf t}({Q\mathstrut}_{0}{W\mathstrut}_{\!0}, \,{\widetilde{Q}\mathstrut}_{0}{\,\widetilde{\!W}\mathstrut}_{\!0}) > {F\mathstrut}_{2}({Q\mathstrut}_{0}{W\mathstrut}_{\!0}, \,R)$ first.
Then ${F\mathstrut}^{\bf t}({Q\mathstrut}_{0}{W\mathstrut}_{\!0}, \,{\widetilde{Q}\mathstrut}_{0}{\,\widetilde{\!W}\mathstrut}_{\!0}, \, R) = {F\mathstrut}_{1}^{\bf t}({Q\mathstrut}_{0}{W\mathstrut}_{\!0}, \,{\widetilde{Q}\mathstrut}_{0}{\,\widetilde{\!W}\mathstrut}_{\!0})$ by (\ref{eqF}).
%Then ${F\mathstrut}_{1}({U\mathstrut}_{\!0}{W\mathstrut}_{\!0}, \,{Q\mathstrut}_{0})$ cannot be decreased in the vicinity of %${U\mathstrut}_{\!0}{W\mathstrut}_{\!0}$ by (\ref{eqMinProcedure}).
%More specifically,
%observe that the point ${U\mathstrut}_{\!0}{W\mathstrut}_{\!0}$ belongs to a closed convex set of distributions $UW$ with $\text{supp}(UW)\subseteq %\text{supp}({Q\mathstrut}_{0}P)$, satisfying the channel input constraint $\alpha$.
%Observe that the function ${F\mathstrut}_{1}^{\bf t}(QW, \,{\widetilde{Q}\mathstrut}_{0}{\,\widetilde{\!W}\mathstrut}_{\!0})$
%is convex ($\cup$) in ${\cal S}$, while the second function in the %$\max$,
%maximization in (\ref{eqF}),
%${F\mathstrut}_{2}(QW, R) = D( W \, \| \, P \, | \, Q) - I(Q, W)  +  R$, is continuous in ${\cal S}$.
By (\ref{eqMinProcedure}), we conclude that ${F\mathstrut}_{1}^{\bf t}({Q\mathstrut}_{0}{W\mathstrut}_{\!0}, \,{\widetilde{Q}\mathstrut}_{0}{\,\widetilde{\!W}\mathstrut}_{\!0})$
cannot be decreased in the vicinity of $QW = {Q\mathstrut}_{0}{W\mathstrut}_{\!0}$ inside the convex set ${\cal S}$.
Let us define a point inside ${\cal S}$:
\begin{align}
& {Q\mathstrut}^{(\lambda)}(x){W\mathstrut}^{(\lambda)}(y \, | \,x) \; \triangleq
\label{eqConvCombin} \\
& \lambda \hat{Q}(x)\hat{W}(y \, | \,x) \, + \, (1 - \lambda){Q\mathstrut}_{0}(x){W\mathstrut}_{\!0}(y \, | \,x), \;\;\;\;\;\; \lambda \, \in \, (0, 1).
\nonumber
\end{align}
We have that ${Q\mathstrut}^{(\lambda)}{W\mathstrut}^{(\lambda)}\in {\cal S}$,
and the function ${f\mathstrut}_{1}(\lambda) \triangleq {F\mathstrut}_{1}^{\bf t}({Q\mathstrut}^{(\lambda)}{W\mathstrut}^{(\lambda)}, \,{\widetilde{Q}\mathstrut}_{0}{\,\widetilde{\!W}\mathstrut}_{\!0})$
is convex ($\cup$) and differentiable %functions of
w.r.t.
$\lambda \in (0, 1)$.
Since ${f\mathstrut}_{1}(\lambda)$ has to be non-decreasing at $\lambda = 0$, the following condition must hold:
\begin{equation} \label{eqFunc1}
\lim_{\lambda \, \rightarrow \, 0} \frac{d{f\mathstrut}_{1}(\lambda)}{d\lambda} \; \geq \; 0.
\end{equation}
Differentiating ${f\mathstrut}_{1}(\lambda)$, similarly as in the proof of
the ``Pythagorean'' theorem for divergence \cite{CoverThomas} (proved as
``the three points property'' in \cite[Lemma~2]{CsiszarTusnady84}),
we obtain:
\begin{align}
& {F\mathstrut}_{1}^{\bf t}({Q\mathstrut}_{0}{W\mathstrut}_{\!0}, \,{\widetilde{Q}\mathstrut}_{0}{\,\widetilde{\!W}\mathstrut}_{\!0})
+ D(\hat{W} \, \| \, {W\mathstrut}_{\!0} \, | \, \hat{Q})
+ {D\mathstrut}^{\bf t}( \hat{Q}\hat{W} , \, {Q\mathstrut}_{0}{W\mathstrut}_{\!0})
\nonumber \\
&
\leq \;\;
{F\mathstrut}_{1}^{\bf t}(\hat{Q}\hat{W}, \,{\widetilde{Q}\mathstrut}_{0}{\,\widetilde{\!W}\mathstrut}_{\!0}).
\label{eqPyth}
\end{align}
Since %$\text{supp}(\hat{Q}\hat{W}) \subseteq \text{supp}(\hat{Q}P)
%\cap \text{supp}({\widetilde{Q}\mathstrut}_{0}{\,\widetilde{\!W}\mathstrut}_{\!0})$, we have %also
${F\mathstrut}_{1}^{\bf t}(\hat{Q}\hat{W}, \,{\widetilde{Q}\mathstrut}_{0}{\,\widetilde{\!W}\mathstrut}_{\!0}) < +\infty$,
then %by (\ref{eqPyth})
%it holds that
the divergences on %its
the LHS of (\ref{eqPyth}) are also finite.
By the definition (\ref{eqF1}),
\begin{equation} \label{eqIdent}
{F\mathstrut}_{1}^{\bf t}(\hat{Q}\hat{W}, \,{\widetilde{Q}\mathstrut}_{0}{\,\widetilde{\!W}\mathstrut}_{\!0}) \; = \;
{F\mathstrut}_{1}^{\bf t}(\hat{Q}\hat{W}, \,\hat{Q}\hat{W})
+ {D\mathstrut}^{\bf t}( \hat{Q}\hat{W} , \, {\widetilde{Q}\mathstrut}_{0}{\,\widetilde{\!W}\mathstrut}_{\!0}).
\end{equation}
Omitting $D(\hat{W} \, \| \, {W\mathstrut}_{\!0} \, | \, \hat{Q}) \geq 0$ from (\ref{eqPyth}),
noting that ${Q\mathstrut}_{0}{W\mathstrut}_{\!0} = {\widetilde{Q}\mathstrut}_{1}{\,\widetilde{\!W}\mathstrut}_{\!1}$, and combining (\ref{eqPyth}) with (\ref{eqIdent}),
we get
\begin{align}
& {F\mathstrut}_{1}^{\bf t}({Q\mathstrut}_{0}{W\mathstrut}_{\!0}, \,{\widetilde{Q}\mathstrut}_{0}{\,\widetilde{\!W}\mathstrut}_{\!0})
\;\; \leq \;\; {F\mathstrut}_{1}^{\bf t}(\hat{Q}\hat{W}, \,\hat{Q}\hat{W})
\nonumber \\
&
\;\;\;\;\;\;\;\;\;\;\;\,
+ {D\mathstrut}^{\bf t}( \hat{Q}\hat{W} , \, {\widetilde{Q}\mathstrut}_{0}{\,\widetilde{\!W}\mathstrut}_{\!0})
-
{D\mathstrut}^{\bf t}( \hat{Q}\hat{W} , \, {\widetilde{Q}\mathstrut}_{1}{\,\widetilde{\!W}\mathstrut}_{\!1}).
\label{eqAlmost}
\end{align}
Now, (\ref{eqBound1}) follows because ${F\mathstrut}^{\bf t}({Q\mathstrut}_{0}{W\mathstrut}_{\!0}, \,{\widetilde{Q}\mathstrut}_{0}{\,\widetilde{\!W}\mathstrut}_{\!0}, \, R) = \newline {F\mathstrut}_{1}^{\bf t}({Q\mathstrut}_{0}{W\mathstrut}_{\!0}, \,{\widetilde{Q}\mathstrut}_{0}{\,\widetilde{\!W}\mathstrut}_{\!0})$
and ${F\mathstrut}_{1}^{\bf t}(\hat{Q}\hat{W}, \,\hat{Q}\hat{W}) \leq {F\mathstrut}^{\bf t}(\hat{Q}\hat{W}, \hat{Q}\hat{W}, R)$.

Consider the case ${F\mathstrut}_{1}^{\bf t}({Q\mathstrut}_{0}{W\mathstrut}_{\!0}, \,{\widetilde{Q}\mathstrut}_{0}{\,\widetilde{\!W}\mathstrut}_{\!0}) < {F\mathstrut}_{2}({Q\mathstrut}_{0}{W\mathstrut}_{\!0}, \,R)$ next.
Then ${F\mathstrut}^{\bf t}({Q\mathstrut}_{0}{W\mathstrut}_{\!0}, \,{\widetilde{Q}\mathstrut}_{0}{\,\widetilde{\!W}\mathstrut}_{\!0}, \, R) = {F\mathstrut}_{2}({Q\mathstrut}_{0}{W\mathstrut}_{\!0}, \,R)$ by (\ref{eqF}).
By (\ref{eqMinProcedure}), we conclude that ${F\mathstrut}_{2}({Q\mathstrut}_{0}{W\mathstrut}_{\!0}, \,R)$
cannot be decreased in the vicinity of $QW = {Q\mathstrut}_{0}{W\mathstrut}_{\!0}$ inside the convex set ${\cal S}$, and by convexity ($\cup$)
of ${F\mathstrut}_{2}(QW, R)$ it follows that
\begin{align}
{F\mathstrut}_{2}({Q\mathstrut}_{0}{W\mathstrut}_{\!0}, \,R) \;\; & = \;\;
\min_{\substack{\\QW \in \, {\cal S}
}} {F\mathstrut}_{2}(QW, R)
\nonumber \\
& \overset{(a)}{\leq} \;\; {F\mathstrut}_{2}(\hat{Q}\hat{W}, R) \;\; \overset{(b)}{\leq} \;\; {F\mathstrut}^{\bf t}(\hat{Q}\hat{W}, \,\hat{Q}\hat{W}, \,R),
\nonumber
\end{align}
where ($a$) follows because $\hat{Q}\hat{W}\in {\cal S}$, %by (\ref{eqAchieve}),
and ($b$) follows by (\ref{eqF}). This again gives (\ref{eqBound1}).

Finally, assume now the equality ${F\mathstrut}_{1}^{\bf t}({Q\mathstrut}_{0}{W\mathstrut}_{\!0}, \,{\widetilde{Q}\mathstrut}_{0}{\,\widetilde{\!W}\mathstrut}_{\!0}) = {F\mathstrut}_{2}({Q\mathstrut}_{0}{W\mathstrut}_{\!0}, \,R)$.
In this case, using the definition (\ref{eqConvCombin}), we
look at two functions: ${f\mathstrut}_{1}(\lambda)$ and
${f\mathstrut}_{2}(\lambda) \triangleq {F\mathstrut}_{2}({Q\mathstrut}^{(\lambda)}{W\mathstrut}^{(\lambda)}, R)$,
both of which are convex ($\cup$) and differentiable w.r.t. $\lambda \in (0, 1)$.
At least one of these two functions
has to be non-decreasing at $\lambda = 0$.
This implies either (\ref{eqFunc1}) or
\begin{equation} \label{eqFunc2}
\lim_{\lambda \, \rightarrow \, 0} \frac{d{f\mathstrut}_{2}(\lambda)}{d\lambda} \; \geq \; 0.
\end{equation}
The condition (\ref{eqFunc1}) results in (\ref{eqBound1}) as before,
while (\ref{eqFunc2}) by convexity ($\cup$) of ${f\mathstrut}_{2}(\lambda)$ implies
\begin{displaymath}
{F\mathstrut}_{2}({Q\mathstrut}_{0}{W\mathstrut}_{\!0}, R) \; \leq \; {F\mathstrut}_{2}(\hat{Q}\hat{W}, R)
\; \leq \; {F\mathstrut}^{\bf t}(\hat{Q}\hat{W}, \hat{Q}\hat{W}, R),
\end{displaymath}
where the second inequality is by definition (\ref{eqF}). Since
${F\mathstrut}_{2}({Q\mathstrut}_{0}{W\mathstrut}_{\!0}, \,R) = {F\mathstrut}^{\bf t}({Q\mathstrut}_{0}{W\mathstrut}_{\!0}, \,{\widetilde{Q}\mathstrut}_{0}{\,\widetilde{\!W}\mathstrut}_{\!0}, \, R)$,
this gives (\ref{eqBound1}).
\end{proof}

\bigskip

A similar, alternative, lemma can be proved if we add ${D\mathstrut}^{\bf t}(QW, \,\widetilde{Q}\,\widetilde{\!W})$
to the second term of the maximum in (\ref{eqEquality}), and not to the first.

\bigskip

{\em Proof of Theorem~\ref{thmConvergence}:} %\newline
By (\ref{eqEquality}) we can rewrite the RHS of (\ref{eqConvergence}) as
\begin{equation} \label{eqAchieve}
\min_{\substack{\\\widetilde{Q}, \, \,\widetilde{\!W}:\\
{D\mathstrut}^{\bf t}(\widetilde{Q}\,\widetilde{\!W}\!, \,{\widetilde{Q}\mathstrut}_{0}{\,\widetilde{\!W}\mathstrut}_{\!0})
\,<\, \infty}}
\!\!\!\!\!\!\!\!\!\!\!\!\!
{E\mathstrut}_{\!c}^{\bf t}(\widetilde{Q}\,\widetilde{\!W}, R, \alpha)
\; = \!\!
\min_{\substack{\\ Q, \, W:\\
\mathbb{E}_{Q}[f(X)]\;\leq \; \alpha
\\ {D\mathstrut}^{\bf t}(QW, \,{\widetilde{Q}\mathstrut}_{0}{\,\widetilde{\!W}\mathstrut}_{\!0}) \,<\, \infty
}}
\!\!\!\!\!\!\!\!\!\!\!\!\!
{F\mathstrut}^{\bf t}(QW, QW, R).
\end{equation}
Suppose (\ref{eqAchieve}) is finite, and let $\hat{Q}\hat{W}$ achieve the RHS $\min$ in (\ref{eqAchieve}). Then
${F\mathstrut}_{1}^{\bf t}(\hat{Q}\hat{W}, \,{\widetilde{Q}\mathstrut}_{0}{\,\widetilde{\!W}\mathstrut}_{\!0}) < +\infty$
and $\,\sum_{x}\hat{Q}(x)f(x) \leq \alpha$.
Then Lemma~\ref{lem5PP} implies that there exist only two possibilities for the outcome of the iterations in (\ref{eqMinProcedure}).
One possibility is that at some iteration $\ell$ it holds that
\begin{displaymath}
{F\mathstrut}^{\bf t}({Q\mathstrut}_{\ell}{W\mathstrut}_{\!\ell}, \,{\widetilde{Q}\mathstrut}_{\ell}{\,\widetilde{\!W}\mathstrut}_{\!\ell}, \, R)
\;\; \leq \;\;
{F\mathstrut}^{\bf t}(\hat{Q}\hat{W}, \hat{Q}\hat{W}, R),
\end{displaymath}
meaning that the monotonically non-increasing sequence of
${F\mathstrut}^{\bf t}({Q\mathstrut}_{\ell}\,{W\mathstrut}_{\!\ell}, \,{\widetilde{Q}\mathstrut}_{\ell}{\,\widetilde{\!W}\mathstrut}_{\!\ell}, \, R) = {E\mathstrut}_{\!c}^{\bf t}({\widetilde{Q}\mathstrut}_{\ell}{\,\widetilde{\!W}\mathstrut}_{\!\ell}, R, \alpha)$
has converged to (\ref{eqAchieve}).
The alternative possibility is that for {\em all} iterations $\ell = 0, 1, 2, ... \, ,$
it holds that
\begin{align}
& {F\mathstrut}^{\bf t}({Q\mathstrut}_{\ell}{W\mathstrut}_{\!\ell}, \,{\widetilde{Q}\mathstrut}_{\ell}{\,\widetilde{\!W}\mathstrut}_{\!\ell}, \, R)
\;\; \leq \;\; {F\mathstrut}^{\bf t}(\hat{Q}\hat{W}, \hat{Q}\hat{W}, R)
\nonumber \\
&
\;\;\;\;\;\;\;\;\;\;\;\;\;\,
+ \, {F\mathstrut}_{1}^{\bf t}( \hat{Q}\hat{W} , \, {\widetilde{Q}\mathstrut}_{\ell}{\,\widetilde{\!W}\mathstrut}_{\!\ell})  \, - \,   {F\mathstrut}_{1}^{\bf t}( \hat{Q}\hat{W} , \, {\widetilde{Q}\mathstrut}_{\ell \, + \, 1}{\,\widetilde{\!W}\mathstrut}_{\!\ell \, + \, 1}),
\nonumber
\end{align}
with all terms finite.
Now, just like in \cite[Lemma~1]{CsiszarTusnady84}, it has to be true %cannot be true
that
\begin{displaymath}
\liminf_{\ell \, \rightarrow \, \infty} \,\Big\{{F\mathstrut}_{1}^{\bf t}( \hat{Q}\hat{W} , \, {\widetilde{Q}\mathstrut}_{\ell}{\,\widetilde{\!W}\mathstrut}_{\!\ell})  \, - \,   {F\mathstrut}_{1}^{\bf t}( \hat{Q}\hat{W} , \, {\widetilde{Q}\mathstrut}_{\ell \, + \, 1}{\,\widetilde{\!W}\mathstrut}_{\!\ell \, + \, 1})\Big\} \; \leq \; 0,
\end{displaymath}
because the divergences in (\ref{eqF1}) are non-negative (i.e., bounded from below). Therefore ${F\mathstrut}^{\bf t}({Q\mathstrut}_{\ell}{W\mathstrut}_{\!\ell}, \,{\widetilde{Q}\mathstrut}_{\ell}{\,\widetilde{\!W}\mathstrut}_{\!\ell}, \, R)$
must converge to ${F\mathstrut}^{\bf t}(\hat{Q}\hat{W}, \hat{Q}\hat{W}, R)$, yielding
(\ref{eqAchieve}), and this concludes the proof of Theorem~\ref{thmConvergence}.
$\square$

\bigskip

\section{Convergence of the iterative minimization for fixed gradient}\label{FixedSlope}

\bigskip

Let us define for two real numbers $0 \leq \rho < 1$ and $\eta \geq 0$
\begin{align}
& {F\mathstrut}^{\bf t}(\rho, \,\eta, \,QW, \,\widetilde{Q}\,\widetilde{\!W}) \;\; \triangleq \;\; D(W \, \| \, P \, | \, Q)
\, - \, \rho \,I(Q, W)
\nonumber \\
&
\;\;\;\;\;\;\;\;\;\;\;\;\;\;\;\,\,
+ \, \eta \,\mathbb{E}_{Q}[f(X)]
\, + \, (1 - \rho){D\mathstrut}^{\bf t}(QW, \,\widetilde{Q}\,\widetilde{\!W}),
\label{eqFixed} \\
%\end{align}
%where the expectaion $\mathbb{E}\,[f(X)]$ is according to $U(x)$.
%\begin{equation}
&
\;\;\;\;\;\;\;\;\,\,
{E\mathstrut}_{0}^{\bf t}(\rho, \eta, \widetilde{Q}\,\widetilde{\!W}) \;\; \triangleq \;\;
\min_{\substack{\\Q, \, W
}}\;
{F\mathstrut}^{\bf t}(\rho, \,\eta, \,QW, \,\widetilde{Q}\,\widetilde{\!W}).
\label{eqE0}
\end{align}
If finite, the quantity ${E\mathstrut}_{0}^{\bf t}(\rho, \eta, \widetilde{Q}\,\widetilde{\!W})$ has a meaning of the vertical axis intercept (``${E\mathstrut}_{0}$'') of a lower supporting plane in the variables $(R, \alpha)$ for the function $E(R, \alpha) = {E\mathstrut}_{\!c}^{\bf t}(\widetilde{Q}\,\widetilde{\!W}, R, \alpha)$, defined in (\ref{eqDef}),
%which is a surface in variables $(R, \alpha)$, as defined by (\ref{eqDef}).
as the following lemma shows.

\bigskip

\begin{lemma} %[]
\label{lemSuppLine}%\newline
{\em For any $0 \leq \rho < 1$ and $\eta \geq 0$ it holds that}
\begin{equation} \label{eqEquiv}
{E\mathstrut}_{\!c}^{\bf t}(\widetilde{Q}\,\widetilde{\!W}, R, \alpha) \;\; \geq \;\;
{E\mathstrut}_{0}^{\bf t}(\rho, \eta, \widetilde{Q}\,\widetilde{\!W}) \, + \, \rho R \, - \, \eta \alpha,
\end{equation}
{\em and there exist $R \geq 0$ and $\alpha \geq \min_{\,x} f(x)$
which satisfy (\ref{eqEquiv}) with equality.}
%that there is equality in (\ref{eqEquiv}).}
\end{lemma}

\bigskip

\begin{proof}
By definition (\ref{eqDef})
%(\ref{eqDCcordecexp})-(\ref{eqInequality}) we can write ${E\mathstrut}_{\!c}(R, \alpha)$ as (\ref{eqBasis}):
\begin{align}
&
\;\;\;\;\;\,\,
\min_{\substack{\\Q, \, W:\\
\mathbb{E}_{Q}[f(X)]\;\leq \; \alpha
}}
\;
\Big\{
D( W \, \| \, P \, | \, Q) \, + \, {D\mathstrut}^{\bf t}(QW, \,\widetilde{Q}\,\widetilde{\!W}) \, +
\nonumber \\
%\end{align}
%\begin{align}
&
\;\;
\big|R - I(Q, W) - {D\mathstrut}^{\bf t}(QW, \,\widetilde{Q}\,\widetilde{\!W})\big|^{+}
\Big\}
\label{eqByDef} \\
%\end{align}
%\begin{align}
&  \overset{(a)}{\geq}  \;
\min_{\substack{\\Q, \, W:\\
\mathbb{E}_{Q}[f(X)]\;\leq \; \alpha
}}
\;
\Big\{
D( W \, \| \, P \, | \, Q) \, + \, {D\mathstrut}^{\bf t}(QW, \,\widetilde{Q}\,\widetilde{\!W}) \, +
\nonumber \\
&
\rho\big[R - I(Q, W)- {D\mathstrut}^{\bf t}(QW, \,\widetilde{Q}\,\widetilde{\!W})\big]
\, + \, \eta\big[\mathbb{E}_{Q}[f(X)] - \alpha\big]\Big\},
\nonumber \\
&  \geq  \;
\;\;\;\;\;\;\,
\min_{\substack{\\Q, \, W
}}
\;\;\;\;\;\;\,
\Big\{
D( W \, \| \, P \, | \, Q) \, + \, {D\mathstrut}^{\bf t}(QW, \,\widetilde{Q}\,\widetilde{\!W}) \, +
\nonumber \\
&
\rho\big[R - I(Q, W)- {D\mathstrut}^{\bf t}(QW, \,\widetilde{Q}\,\widetilde{\!W})\big]
\, + \, \eta\big[\mathbb{E}_{Q}[f(X)] - \alpha\big]\Big\},
\label{eqConvexEnv}
\end{align}
where
($a$) holds for any $0 \leq \rho < 1$ and $\eta \geq 0$.
Using (\ref{eqFixed}) and (\ref{eqE0}),
we see that
the lower bound expression (\ref{eqConvexEnv}) is equal to the RHS of (\ref{eqEquiv}). %by (\ref{eqFixed}), (\ref{eqE0}).
Suppose (\ref{eqConvexEnv}) is finite.
Let ${Q\mathstrut}_{\rho, \, \eta}$, ${W\mathstrut}_{\!\rho, \, \eta}$ %${Q\mathstrut}_{\rho, \, \eta}$
 denote distributions $Q$, $W$, respectively, which jointly minimize (\ref{eqConvexEnv}).
% Then it can be seen from (\ref{eqFixed}), that necessarily ${Q\mathstrut}_{\rho, \, \eta} = {U\mathstrut}_{\!\rho, \, \eta}$.
Observe that for each $0 \leq \rho < 1$ and $\eta \geq 0$ we can find $R \geq 0$ and $\alpha \geq \min_{\,x} f(x)$,
such that the differences in the square brackets are zero.
In this case, ${Q\mathstrut}_{\rho, \, \eta}$ %and ${Q\mathstrut}_{\rho, \, \eta}$
will satisfy the input constraint and
there will be equality between (\ref{eqConvexEnv}) and (\ref{eqByDef}).
\end{proof}

%\bigskip

%%%%%%%%%%%%%%%%%%%%%%%%%%%%%%%%%%%%%%%%%%
%%% REMARK ABOUT CONVEXITY - NEED TO CHECK
%%%%%%%%%%%%%%%%%%%%%%%%%%%%%%%%%%%%%%%%%%
%In fact, since ${E\mathstrut}_{\!c}^{\bf t}(\widetilde{Q}\,\widetilde{\!W}, R, \alpha)$ is a convex ($\cup$) and monotonic function of $(R, \alpha)$,
%which cannot have lower supporting planes with slopes $\rho > 1$, the supremum of the RHS of (\ref{eqEquiv}) over $0 \leq \rho < 1$ and $\eta \geq 0$
%equals ${E\mathstrut}_{\!c}^{\bf t}(\widetilde{Q}\,\widetilde{\!W}, R, \alpha)$ for all $(R, \alpha)$.
%%%%%%%%%%%%%%%%%%%%%%%%%%%%%%%%%%%%%%%%%%

%Since (\ref{eqByDef}) is a monotonic function of $R$ and $\alpha$, and  in the limit of large $R$ the slope of (\ref{eqByDef}), as a function of $R$, %cannot be greater than $1$,
%we conclude that (\ref{eqConvexEnv}) is the lower convex envelope of (\ref{eqByDef}) as a function of $(R, \alpha)$.
%Observe now, that (\ref{eqByDef}) itself is a convex ($\cup$) function of $(R, \alpha)$.
%This can be checked directly using definition (\ref{eqDef}) and the fact, that $F(UW, Q, R)$
%is a convex ($\cup$) function of the pair $(UW, R)$.
%The equality in (\ref{eqEquiv}) now follows, because a convex ($\cup$) function coincides with its lower convex envelope.

\bigskip

\begin{lemma} %[]
\label{lemSol}%\newline
{\em Suppose $\widetilde{Q}\,\widetilde{\!W} \equiv \widetilde{T}\widetilde{V}$ is such that the minimum (\ref{eqE0}) is finite. If $t_{1} = t_{4} + 1$ in (\ref{eqCombination}),
then, with definitions of $a \triangleq (t_{2} + t_{4})(1-\rho)$ and $b \triangleq (t_{3} + t_{4})(1-\rho)$,
$0\leq \rho < 1$ and $\eta \geq 0$,
the unique minimizing solution of the minimum (\ref{eqE0}) can be written as}
%inner minimum from (\ref{eqEquiv})}
\begin{align}
& {Q\mathstrut}^{*}(x){W\mathstrut}^{*}(y \, | \, x) \;\; = \;\;
\frac{1}{K}\Big[{\widetilde{Q}\mathstrut}^{1-\rho}(x){\widetilde{V}\mathstrut}^{b}(x\,|\,y)P_{\eta}(x, y)\Big]^{\frac{1}{b+1-\rho}}
\nonumber \\
&
\times{\widetilde{T}\mathstrut}^{\frac{a}{a+1}}(y)
\bigg\{\sum_{\tilde{x}}\Big[{\widetilde{Q}\mathstrut}^{1-\rho}(\tilde{x}){\widetilde{V}\mathstrut}^{b}(\tilde{x}\,|\,y)P_{\eta}(\tilde{x}, y)\Big]^{\frac{1}{b+1-\rho}}\bigg\}^{\frac{b-a-\rho}{a+1}},
\label{eqSol}
\end{align}
{\em where $\,P_{\eta}(x, y) \triangleq
e^{-\eta f(x)} P(y\,|\,x)$ and $K$ is a normalization constant, resulting in}
\begin{align}
& {E\mathstrut}_{0}^{\bf t}(\rho, \eta, \widetilde{Q}\,\widetilde{\!W}) \;\; = \;\;
-(a+1)\log\sum_{y}{\widetilde{T}\mathstrut}^{\frac{a}{a+1}}(y)\times
\nonumber \\
& \bigg\{\sum_{x}\Big[{\widetilde{Q}\mathstrut}^{1-\rho}(x){\widetilde{V}\mathstrut}^{b}(x\,|\,y)P_{\eta}(x, y)\Big]^{\frac{1}{b+1-\rho}}\bigg\}^{\frac{b+1-\rho}{a+1}}.
\label{eqE0Explicit}
\end{align}
{\em If $t_{3} = t_{2} + \frac{\rho}{1-\rho}$ in (\ref{eqCombination}),
then, with $c \triangleq (t_{1} + t_{2})(1-\rho)$ and $a$ as defined above,
$0 < \rho < 1$ and $\eta \geq 0$,
the unique minimizing solution of the minimum (\ref{eqE0}) can be written as}
\begin{align}
& {Q\mathstrut}^{*}(x){W\mathstrut}^{*}(y \, | \, x) \;\; = \;\;
\frac{1}{K}\Big[{\,\widetilde{\!W}\mathstrut}^{a}(y\,|\,x){\widetilde{V}\mathstrut}^{\rho}(x\,|\,y)P_{\eta}(x, y)\Big]^{\frac{1}{a+1}}
\nonumber \\
&
\times{\widetilde{Q}\mathstrut}^{\frac{c}{c+\rho}}(x)
\bigg\{\sum_{\tilde{y}}\Big[{\,\widetilde{\!W}\mathstrut}^{a}(\tilde{y}\,|\,x){\widetilde{V}\mathstrut}^{\rho}(x\,|\,\tilde{y})P_{\eta}(x, \tilde{y})\Big]^{\frac{1}{a+1}}\bigg\}^{\frac{a+1-c-\rho}{c+\rho}}\!\!,
\label{eqSol2}
\end{align}
{\em where $P_{\eta}(x, y)$ is defined as above and $K$ is a normalization constant, resulting in}
\begin{align}
& {E\mathstrut}_{0}^{\bf t}(\rho, \eta, \widetilde{Q}\,\widetilde{\!W}) \;\; = \;\;
-(c+\rho)\log\sum_{x}{\widetilde{Q}\mathstrut}^{\frac{c}{c+\rho}}(x)\times
\nonumber \\
& \bigg\{\sum_{y}\Big[{\,\widetilde{\!W}\mathstrut}^{a}(y\,|\,x){\widetilde{V}\mathstrut}^{\rho}(x\,|\,y)P_{\eta}(x, y)\Big]^{\frac{1}{a+1}}\bigg\}^{\frac{a+1}{c+\rho}}.
\label{eqE0Explicit2}
\end{align}
%{\em where $\,P_{\eta}(x, y) \triangleq
%e^{-\eta f(x)} P(y\,|\,x)$ and $K$ is a constant.}
\end{lemma}
\begin{proof} Similarly to \cite[Lemma~3]{TridenskiZamir18Trans}.
\end{proof}

\bigskip

An iterative minimization procedure at a fixed gradient $(\rho, \eta)$,
$0 < \rho < 1$, $\eta \geq 0$,
%uses the explicit computation of (\ref{eqSol}) and
is given by
\begin{align}
&
\begin{array}{l}
\displaystyle
\;\;\;\;\;\;\;\,
{Q\mathstrut}_{\ell}{W\mathstrut}_{\!\ell} \;\; = \;\;
\underset{\substack{\\Q, \, W
}}{\arg\min}
%\arg\min_{\substack{\\U(x), \, W(y\,|\,x)
%}}
\;
{F\mathstrut}^{\bf t}(\rho, \,\eta, \,QW, \,{\widetilde{Q}\mathstrut}_{\ell}{\,\widetilde{\!W}\mathstrut}_{\!\ell}), \\
%F(\rho, \eta, UW, {Q\mathstrut}_{\ell}), \\ %\;\;\; \ell \, = \, 0, 1, 2, ...
%\label{eqStep1} \\
\displaystyle
{\widetilde{Q}\mathstrut}_{\ell\,+\,1}{\,\widetilde{\!W}\mathstrut}_{\!\ell\,+\,1} \; = \;\;\;
{Q\mathstrut}_{\ell}{W\mathstrut}_{\!\ell},
\end{array}
\label{eqMinProcGradient} \\
&
\;\;\;\;\;\;\;\;\;\;\;\;\;\;\;\;\;\;\;\;\;\;\;\;\;\;\;\;\;\;\;\;\;\;\;\;\;\;\;\;\;\;\;\;
\ell \, = \, 0, 1, 2, ... \,.
\nonumber
\end{align}
We assume that the initial distribution ${\widetilde{Q}\mathstrut}_{0}{\,\widetilde{\!W}\mathstrut}_{\!0}$ in (\ref{eqMinProcGradient}) is chosen such that
the set $\big\{QW:  {F\mathstrut}_{1}^{\bf t}(QW, \,{\widetilde{Q}\mathstrut}_{0}{\,\widetilde{\!W}\mathstrut}_{\!0}) < +\infty \big\}$
is non-empty, which guarantees ${F\mathstrut}^{\bf t}(\rho, \, \eta, \,{Q\mathstrut}_{0}{W\mathstrut}_{\!0}, \,{\widetilde{Q}\mathstrut}_{0}{\,\widetilde{\!W}\mathstrut}_{\!0}) = {E\mathstrut}_{0}^{\bf t}(\rho, \eta, {\widetilde{Q}\mathstrut}_{0}{\,\widetilde{\!W}\mathstrut}_{\!0}) < +\infty$.
By (\ref{eqFixed}) it is clear that (\ref{eqMinProcGradient}) produces a monotonically non-increasing sequence ${E\mathstrut}_{0}^{\bf t}(\rho, \eta, {\widetilde{Q}\mathstrut}_{\ell}{\,\widetilde{\!W}\mathstrut}_{\!\ell})$, $\ell = 0, 1, 2, ... \,$.
Depending on the choice of the non-negative parameters $(t_{1}, t_{2}, t_{3}, t_{4})$ in (\ref{eqCombination}),
the update of ${Q\mathstrut}_{\ell}\,{W\mathstrut}_{\!\ell}$ in (\ref{eqMinProcGradient}) can be done according to the expression (\ref{eqSol})
with any $a\geq 0$ and $b \geq 0$, or according to (\ref{eqSol2}) with any $a\geq 0$ and $c \geq 0$,
with $\widetilde{Q}$, $\widetilde{V}$, $\widetilde{T}$, $\,\widetilde{\!W}$ replaced by ${\widetilde{Q}\mathstrut}_{\ell}$,
${\widetilde{V}\mathstrut}_{\ell}$, ${\widetilde{T}\mathstrut}_{\ell}$, ${\,\widetilde{\!W}\mathstrut}_{\!\ell}$,
correspondingly.
The choice of $a = b = 0$ in (\ref{eqSol}) gives the fixed-slope counterpart of the algorithm in \cite{TridenskiZamir18Trans},
analysed in \cite{TridenskiZamir18Extended}. The choice $(a, c) = (0, 1)$ in (\ref{eqSol2})
gives the fixed-slope counterpart of the algorithm in \cite{TridenskiZamir18}.
The choice $(a, b) = (0, \rho)$ in (\ref{eqSol}), or, alternatively, $(a, c) = (0, 1-\rho)$ in (\ref{eqSol2})
gives the algorithm in \cite{OohamaJitsumatsu19}, \cite{OohamaJitsumatsu15}.
The main result of the section is given by the following theorem:

\bigskip

\begin{thm}%[]
\label{thmConvergeGradient}%\newline
{\em Let ${\big\{{Q\mathstrut}_{\ell}{W\mathstrut}_{\!\ell}\big\}\mathstrut}_{\ell \, = \, 0}^{+\infty}$ be a sequence of iterative solutions produced by (\ref{eqMinProcGradient}). %\newline
%If $\,{E\mathstrut}_{\!c}({Q\mathstrut}_{0}, R, \alpha) < +\infty$, %\newline
Then}
\begin{equation} \label{eqConvergeGrad}
%F(\rho, \,\eta, \,{U\mathstrut}_{\!\ell}\,{W\mathstrut}_{\!\ell}, \,{Q\mathstrut}_{\ell})
{E\mathstrut}_{0}^{\bf t}(\rho, \eta, {\widetilde{Q}\mathstrut}_{\ell}{\,\widetilde{\!W}\mathstrut}_{\!\ell})
\; \overset{\ell \, \rightarrow\,\infty}{\searrow} \;
\min_{\substack{\\\widetilde{Q}, \, \,\widetilde{\!W}:\\ {D\mathstrut}^{\bf t}(\widetilde{Q}\,\widetilde{\!W}\!, \,{\widetilde{Q}\mathstrut}_{0}{\,\widetilde{\!W}\mathstrut}_{\!0})
\,<\, \infty}}
%F(\rho, \eta, UW, U),
{E\mathstrut}_{0}^{\bf t}(\rho, \eta, \widetilde{Q}\,\widetilde{\!W}),
%\;\;
%\min_{\substack{\\U(x), \, W(y\,|\,x):\\
%\mathbb{E}\,[f(X)]\;\leq \; \alpha
%}}
%\;
%F(UW, Q).
\end{equation}
{\em where ${E\mathstrut}_{0}^{\bf t}(\rho, \eta, \widetilde{Q}\,\widetilde{\!W})$ is defined in (\ref{eqE0}) and ${D\mathstrut}^{\bf t}(\cdot\,,\cdot)$ in (\ref{eqCombination}).}
\end{thm}

\bigskip

In order to prove Theorem~\ref{thmConvergeGradient}, we use the following lemma:

\bigskip

\begin{lemma} %[]
\label{lem5PPFixed}%\newline
{\em Let $\hat{Q}\hat{W}$ be such that
${F\mathstrut}_{1}^{\bf t}(\hat{Q}\hat{W}, \,{\widetilde{Q}\mathstrut}_{0}{\,\widetilde{\!W}\mathstrut}_{\!0}) < +\infty$.
%$\,\text{\em supp}(\hat{U}\hat{W}) \subseteq \text{\em supp}({Q\mathstrut}_{0}P)$.
%achieve the minimum in (\ref{eqAchieve}). %and suppose $\,{E\mathstrut}_{\!c}({Q\mathstrut}_{0}, R, \alpha) < +\infty$.
%and let $\text{\em supp}({Q\mathstrut}_{1}) = \text{\em supp}({Q\mathstrut}_{0})$.}\newline
Then} %$\text{\em supp}(\hat{U}) \subseteq \text{\em supp}({Q\mathstrut}_{1})$ and}
\begin{align}
& {F\mathstrut}^{\bf t}(\rho, \,\eta, \,{Q\mathstrut}_{0}{W\mathstrut}_{\!0}, \,{\widetilde{Q}\mathstrut}_{0}{\,\widetilde{\!W}\mathstrut}_{\!0})
\;\; \leq \;\;
{F\mathstrut}^{\bf t}(\rho, \eta, \hat{Q}\hat{W}, \hat{Q}\hat{W})
\nonumber \\
&
\;\;\;
 + \, %F(\rho, \eta, \hat{U}\hat{W}, {Q\mathstrut}_{0})
(1 - \rho)\Big[{F\mathstrut}_{1}^{\bf t}(\hat{Q}\hat{W}, \,{\widetilde{Q}\mathstrut}_{0}{\,\widetilde{\!W}\mathstrut}_{\!0})
\, - \, %F(\rho, \eta, \hat{U}\hat{W}, {Q\mathstrut}_{1}).
{F\mathstrut}_{1}^{\bf t}(\hat{Q}\hat{W}, \,{\widetilde{Q}\mathstrut}_{1}{\,\widetilde{\!W}\mathstrut}_{\!1})\Big].
\label{eqBoundFixed}
\end{align}
\end{lemma}
%\bigskip
\begin{proof}
Since $+\infty > {F\mathstrut}_{1}^{\bf t}(\hat{Q}\hat{W}, \,{\widetilde{Q}\mathstrut}_{0}{\,\widetilde{\!W}\mathstrut}_{\!0})$,
then also $+\infty > {F\mathstrut}_{1}^{\bf t}({Q\mathstrut}_{0}{W\mathstrut}_{\!0}, \,{\widetilde{Q}\mathstrut}_{0}{\,\widetilde{\!W}\mathstrut}_{\!0})$.
Let ${Q\mathstrut}^{(\lambda)}{W\mathstrut}^{(\lambda)}$ be a convex combination of $\hat{Q}\hat{W}$
and ${Q\mathstrut}_{0}{W\mathstrut}_{\!0}$, as in (\ref{eqConvCombin}).
Then the function $g(\lambda) = {F\mathstrut}^{\bf t}(\rho, \,\eta, \,{Q\mathstrut}^{(\lambda)}{W\mathstrut}^{(\lambda)}, \,{\widetilde{Q}\mathstrut}_{0}{\,\widetilde{\!W}\mathstrut}_{\!0})$
is convex ($\cup$) and differentiable in $\lambda \in (0, 1)$.
Since ${Q\mathstrut}_{0}{W\mathstrut}_{\!0}$ achieves the minimum of ${F\mathstrut}^{\bf t}(\rho, \,\eta, \,QW, \,{\widetilde{Q}\mathstrut}_{0}{\,\widetilde{\!W}\mathstrut}_{\!0})$ over $QW$,
then necessarily
\begin{displaymath}
\lim_{\lambda \, \rightarrow \, 0} \frac{d g(\lambda)}{d\lambda} \; \geq \; 0.
\end{displaymath}
Differentiation results in the following condition in the limit:
\begin{align}
&
\;\;\;
{F\mathstrut}^{\bf t}(\rho, \,\eta, \,{Q\mathstrut}_{0}{W\mathstrut}_{\!0}, \,{\widetilde{Q}\mathstrut}_{0}{\,\widetilde{\!W}\mathstrut}_{\!0})
\, + \, \rho D(\hat{T}\,\|\,{T\mathstrut}_{0})
\nonumber \\
&
\;\;\;
+\,(1-\rho)\Big[ D(\hat{W}\,\|\,{W\mathstrut}_{\!0} \,|\,\hat{Q})
+ {D\mathstrut}^{\bf t}(\hat{Q}\hat{W}\,\|\,{Q\mathstrut}_{0}{W\mathstrut}_{\!0})\Big]
\nonumber \\
& \leq \; {F\mathstrut}^{\bf t}(\rho, \,\eta, \,\hat{Q}\hat{W}, \,{\widetilde{Q}\mathstrut}_{0}{\,\widetilde{\!W}\mathstrut}_{\!0}),
\label{eqDiffResult}
\end{align}
where $\hat{T}$ and ${T\mathstrut}_{0}$ denote the $y$-marginal distributions of $\hat{Q}\hat{W}$ and ${Q\mathstrut}_{0}{W\mathstrut}_{\!0}$,
respectively. %, over ${\cal Y}$.
Since ${F\mathstrut}_{1}^{\bf t}(\hat{Q}\hat{W}, \,{\widetilde{Q}\mathstrut}_{0}{\,\widetilde{\!W}\mathstrut}_{\!0}) < +\infty$,
then all terms in (\ref{eqDiffResult}) are finite.
%It follows that $D(\hat{U}\hat{W}\,\|\,{U\mathstrut}_{\!0}{W\mathstrut}_{\!0}) < +\infty$ and therefore $\text{supp}(\hat{U}) \subseteq %\text{supp}({Q\mathstrut}_{1})$.
On the other hand, by (\ref{eqFixed})
\begin{align}
& {F\mathstrut}^{\bf t}(\rho, \,\eta, \,\hat{Q}\hat{W}, \,{\widetilde{Q}\mathstrut}_{0}{\,\widetilde{\!W}\mathstrut}_{\!0}) \; =
\nonumber \\
& {F\mathstrut}^{\bf t}(\rho, \,\eta, \,\hat{Q}\hat{W}, \,\hat{Q}\hat{W}) \, + \, (1-\rho) {D\mathstrut}^{\bf t}(\hat{Q}\hat{W}\,\|\,{\widetilde{Q}\mathstrut}_{0}{\,\widetilde{\!W}\mathstrut}_{\!0}).
\label{eqOtherHand}
\end{align}
Combining (\ref{eqOtherHand}) with (\ref{eqDiffResult}),
noting that ${Q\mathstrut}_{0}{W\mathstrut}_{\!0} = {\widetilde{Q}\mathstrut}_{1}{\,\widetilde{\!W}\mathstrut}_{\!1}$,
and omitting non-negative terms $(1-\rho)D(\hat{W}\,\|\,{W\mathstrut}_{\!0} \,|\,\hat{Q}) \geq 0$ and
$\rho D(\hat{T}\,\|\,{T\mathstrut}_{0}) \geq 0$,
%replacing $D(\hat{U}\hat{W}\,\|\,{U\mathstrut}_{\!0}{W\mathstrut}_{\!0})$ with $D(\hat{U}\,\|\,{U\mathstrut}_{\!0})$,
we obtain a weaker inequality (\ref{eqBoundFixed}).
\end{proof}

\bigskip

{\em Proof of Theorem~\ref{thmConvergeGradient}:}
Using (\ref{eqFixed}), (\ref{eqE0}), it can be verified,
that the RHS of (\ref{eqConvergeGrad}) can be rewritten as
\begin{equation} \label{eqRewrite}
\min_{\substack{\\\widetilde{Q}, \, \,\widetilde{\!W}:\\ {D\mathstrut}^{\bf t}(\widetilde{Q}\,\widetilde{\!W}\!, \,{\widetilde{Q}\mathstrut}_{0}{\,\widetilde{\!W}\mathstrut}_{\!0})
\,<\, \infty}}
\!\!\!\!\!\!\!\!\!\!\!\!\!
{E\mathstrut}_{0}^{\bf t}(\rho, \eta, \widetilde{Q}\,\widetilde{\!W}) = \!\!\!\!
\min_{\substack{\\Q, \, W:\\ {D\mathstrut}^{\bf t}(QW, \,{\widetilde{Q}\mathstrut}_{0}{\,\widetilde{\!W}\mathstrut}_{\!0})
\,<\, \infty}}
\!\!\!\!\!\!\!\!\!\!\!\!\!
{F\mathstrut}^{\bf t}(\rho, \eta, QW, QW).
\end{equation}
Suppose (\ref{eqRewrite}) is finite and
let $\hat{Q}\hat{W}$ achieve the minimum on the RHS of (\ref{eqRewrite}). Then by Lemma~\ref{lem5PPFixed}
we conclude that for all iterations $\ell = 0, 1, 2, ... \, ,$
it holds that
\begin{align}
& {F\mathstrut}^{\bf t}(\rho, \,\eta, \,{Q\mathstrut}_{\ell}{W\mathstrut}_{\!\ell}, \,{\widetilde{Q}\mathstrut}_{\ell}{\,\widetilde{\!W}\mathstrut}_{\!\ell})
\;\; \leq \;\;
{F\mathstrut}^{\bf t}(\rho, \eta, \hat{Q}\hat{W}, \hat{Q}\hat{W}) %\; +
\nonumber \\
&
 + \, %F(\rho, \eta, \hat{U}\hat{W}, {Q\mathstrut}_{0})
(1 - \rho)\Big[{F\mathstrut}_{1}^{\bf t}(\hat{Q}\hat{W}, \,{\widetilde{Q}\mathstrut}_{\ell}{\,\widetilde{\!W}\mathstrut}_{\!\ell})
\, - \, %F(\rho, \eta, \hat{U}\hat{W}, {Q\mathstrut}_{1}).
{F\mathstrut}_{1}^{\bf t}(\hat{Q}\hat{W}, \,{\widetilde{Q}\mathstrut}_{\ell\,+\,1}{\,\widetilde{\!W}\mathstrut}_{\!\ell\,+\,1})\Big].
\nonumber
\end{align}
The conclusion of the proof is the same as in Theorem~\ref{thmConvergence}.
$\square$

\bigskip

The next two sections show convergence of fixed-slope computation in the directions of $R$ and $\alpha$, respectively.
They are similar in structure to Section~\ref{FixedSlope}.

\bigskip

\section{Convergence %of iterative minimization
for fixed $\alpha$ and $\rho$}\label{FixedMixed}

\bigskip

%Here
In this section we show convergence of an iterative minimization at a fixed slope $\rho$ in the direction of $R$, i.e., for a given $\alpha$.
With the help of (\ref{eqFixed}) let us define ${F\mathstrut}^{\bf t}(\rho, QW, \widetilde{Q}\,\widetilde{\!W}) \;\triangleq\; \left.{F\mathstrut}^{\bf t}(\rho, \eta, QW, \widetilde{Q}\,\widetilde{\!W})\right|_{\eta \, = \, 0}$ and
%Let us define
%\begin{equation} \label{eqFixedMix}
%F(\rho, UW, Q) \; \triangleq \; D( U W \, \| \, U P) \, + \, (1 - \rho)D(U \, \| \, Q) \, - \, \rho \,I(U, W).
%\end{equation}
%where the expectaion $\mathbb{E}\,[f(X)]$ is according to $U(x)$.
\begin{equation} \label{eqE0Mix}
{E\mathstrut}_{0}^{\bf t}(\rho, \widetilde{Q}\,\widetilde{\!W}, \alpha) \; \triangleq \;
\min_{\substack{\\Q, \, W:\\
\mathbb{E}_{Q}[f(X)]\;\leq \; \alpha
}}
{F\mathstrut}^{\bf t}(\rho, QW, \widetilde{Q}\,\widetilde{\!W}).
\end{equation}
Here ${E\mathstrut}_{0}^{\bf t}(\rho, \widetilde{Q}\,\widetilde{\!W}, \alpha)$ plays a role of ``${E\mathstrut}_{0}$'' of a supporting line in the variable $R$ of the function
$E(R) = {E\mathstrut}_{\!c}^{\bf t}(\widetilde{Q}\,\widetilde{\!W}, R, \alpha)$, defined in (\ref{eqDef}), as shown by the following lemma.

\bigskip

\begin{lemma} %[]
\label{lemSuppLineMix}%\newline
{\em For any $0 \leq \rho < 1$ it holds that}
\begin{equation} \label{eqEquivMix}
{E\mathstrut}_{\!c}^{\bf t}(\widetilde{Q}\,\widetilde{\!W}, R, \alpha) \;\; \geq \;\;
{E\mathstrut}_{0}^{\bf t}(\rho, \widetilde{Q}\,\widetilde{\!W}, \alpha) \, + \, \rho R,
\end{equation}
{\em and there exists %such
$R \geq 0$ which satisfies (\ref{eqEquivMix}) with equality.}
% that there is equality in (\ref{eqEquivMix}).}
\end{lemma}
\begin{proof}
Similar to Lemma~\ref{lemSuppLine}.
\end{proof}

\bigskip

An iterative minimization procedure at a fixed slope $\rho$ is given by
\begin{align}
&
\begin{array}{l}
\displaystyle
\;\;\;\;\;\;\;\;\;
{Q\mathstrut}_{\ell}{W\mathstrut}_{\!\ell} \;\;\, \in \;\;
\underset{\substack{\\Q, \, W:\\
\mathbb{E}_{Q}[f(X)] \;\leq \; \alpha
}}{\arg\min}
{F\mathstrut}^{\bf t}(\rho, \,QW, \,{\widetilde{Q}\mathstrut}_{\ell}{\,\widetilde{\!W}\mathstrut}_{\!\ell}), \\ %\;\;\; \ell \, = \, 0, 1, 2, ...
\displaystyle
{\widetilde{Q}\mathstrut}_{\ell\, + \, 1}{\,\widetilde{\!W}\mathstrut}_{\!\ell\, + \, 1} \;\; = \;\;
\;\;\;\;
{Q\mathstrut}_{\ell}{W\mathstrut}_{\!\ell},
\end{array}
\label{eqMinProcFixMix1} \\
& \;\;\;\;\;\;\;\;\;\;\;\;\;\;\;\;\;\;\;\;\;\;\;\;\;\;\;\;\;\;\;\;\;\;\;\;\;\;\;\;\;\;\;\;\;\;\;\;\,
\ell \, = \, 0, 1, 2, ... \,.
\nonumber
\end{align}
It is assumed that ${\widetilde{Q}\mathstrut}_{0}{\,\widetilde{\!W}\mathstrut}_{\!0}$
in (\ref{eqMinProcFixMix1})
is chosen such that the set
$\big\{QW:  \sum_{x}Q(x)f(x) \leq \alpha, {F\mathstrut}_{1}^{\bf t}(QW, \,{\widetilde{Q}\mathstrut}_{0}{\,\widetilde{\!W}\mathstrut}_{\!0}) < +\infty \big\}$
is non-empty, so that
${F\mathstrut}^{\bf t}(\rho, \,{Q\mathstrut}_{0}{W\mathstrut}_{\!0}, \,{\widetilde{Q}\mathstrut}_{0}{\,\widetilde{\!W}\mathstrut}_{\!0}) = {E\mathstrut}_{\!c}^{\bf t}(\rho, {\widetilde{Q}\mathstrut}_{0}{\,\widetilde{\!W}\mathstrut}_{\!0}, \alpha) < +\infty$.
By the definition of ${F\mathstrut}^{\bf t}(\rho, QW, \widetilde{Q}\,\widetilde{\!W})$
according to (\ref{eqFixed}), this procedure results in a monotonically non-increasing sequence
${E\mathstrut}_{0}^{\bf t}(\rho, {\widetilde{Q}\mathstrut}_{\ell}{\,\widetilde{\!W}\mathstrut}_{\!\ell}, \alpha)$,
$\ell = 0, 1, 2, ... \,$.
The main result of this section is stated in the following theorem.

\bigskip

\begin{thm}%[]
\label{thmConvergeMix1}%\newline
{\em Let ${\big\{{Q\mathstrut}_{\ell}{W\mathstrut}_{\!\ell}\big\}\mathstrut}_{\ell \, = \, 0}^{+\infty}$ be a sequence of iterative solutions produced by (\ref{eqMinProcFixMix1}). %\newline
%If $\,{E\mathstrut}_{\!c}({Q\mathstrut}_{0}, R, \alpha) < +\infty$, %\newline
Then}
\begin{equation} \label{eqConvergeMix1}
%F(\rho, \,\eta, \,{U\mathstrut}_{\!\ell}\,{W\mathstrut}_{\!\ell}, \,{Q\mathstrut}_{\ell})
{E\mathstrut}_{0}^{\bf t}(\rho, {\widetilde{Q}\mathstrut}_{\ell}{\,\widetilde{\!W}\mathstrut}_{\!\ell}, \alpha)
\; \overset{\ell \, \rightarrow\,\infty}{\searrow} \;
\min_{\substack{\\\widetilde{Q}, \, \,\widetilde{\!W}:\\ {D\mathstrut}^{\bf t}(\widetilde{Q}\,\widetilde{\!W}\!, \,{\widetilde{Q}\mathstrut}_{0}{\,\widetilde{\!W}\mathstrut}_{\!0})
\,<\, \infty}}
%F(\rho, \eta, UW, U),
{E\mathstrut}_{0}^{\bf t}(\rho, \widetilde{Q}\,\widetilde{\!W}, \alpha),
\end{equation}
{\em where ${E\mathstrut}_{0}^{\bf t}(\rho, \widetilde{Q}\,\widetilde{\!W}, \alpha)$ is defined in (\ref{eqE0Mix})
and ${D\mathstrut}^{\bf t}(\cdot\,,\cdot)$ in (\ref{eqCombination}).}
\end{thm}

\bigskip

To prove Theorem~\ref{thmConvergeMix1}, we use a lemma, similar to Lemma~\ref{lem5PPFixed}:

\bigskip

\begin{lemma} %[]
\label{lem5PPFixedMixed}%\newline
{\em Let $\hat{Q}\hat{W}$ be such, that $\,\sum_{x}\hat{Q}(x)f(x) \leq \alpha$
and ${F\mathstrut}_{1}^{\bf t}(\hat{Q}\hat{W}, \,{\widetilde{Q}\mathstrut}_{0}{\,\widetilde{\!W}\mathstrut}_{\!0}) < +\infty$.
%achieve the minimum in (\ref{eqAchieve}). %and suppose $\,{E\mathstrut}_{\!c}({Q\mathstrut}_{0}, R, \alpha) < +\infty$.
%and let $\text{\em supp}({Q\mathstrut}_{1}) = \text{\em supp}({Q\mathstrut}_{0})$.}\newline
Then} %$\text{\em supp}(\hat{U}) \subseteq \text{\em supp}({Q\mathstrut}_{1})$ and}
\begin{align}
& {F\mathstrut}^{\bf t}(\rho, \,{Q\mathstrut}_{0}{W\mathstrut}_{\!0}, \,{\widetilde{Q}\mathstrut}_{0}{\,\widetilde{\!W}\mathstrut}_{\!0})
\;\, \leq \;\,
{F\mathstrut}^{\bf t}(\rho, \hat{Q}\hat{W}, \hat{Q}\hat{W})
\nonumber \\
&
\;\;\;
 + \, %F(\rho, \eta, \hat{U}\hat{W}, {Q\mathstrut}_{0})
(1 - \rho)\Big[{F\mathstrut}_{1}^{\bf t}(\hat{Q}\hat{W}, \,{\widetilde{Q}\mathstrut}_{0}{\,\widetilde{\!W}\mathstrut}_{\!0})
\, - \, %F(\rho, \eta, \hat{U}\hat{W}, {Q\mathstrut}_{1}).
{F\mathstrut}_{1}^{\bf t}(\hat{Q}\hat{W}, \,{\widetilde{Q}\mathstrut}_{1}{\,\widetilde{\!W}\mathstrut}_{\!1})\Big].
\label{eqBoundFixedMixed}
\end{align}
\end{lemma}
\begin{proof}
Analogous to Lemma~\ref{lem5PPFixed}.
\end{proof}

\bigskip

{\em Proof of Theorem~\ref{thmConvergeMix1}:}
The RHS of (\ref{eqConvergeMix1}) can be rewritten in terms of ${F\mathstrut}^{\bf t}(\rho, QW, \widetilde{Q}\,\widetilde{\!W})$ as:
\begin{equation} \label{eqRewriteMix1}
\min_{\substack{\\\widetilde{Q}, \, \,\widetilde{\!W}:\\ {D\mathstrut}^{\bf t}(\widetilde{Q}\,\widetilde{\!W}\!, \,{\widetilde{Q}\mathstrut}_{0}{\,\widetilde{\!W}\mathstrut}_{\!0})
\,<\, \infty}}
\!\!\!\!\!\!\!\!\!\!\!\!\!
{E\mathstrut}_{0}^{\bf t}(\rho, \widetilde{Q}\,\widetilde{\!W}, \alpha) \; =
\!\!
\min_{\substack{\\ Q, \, W:\\
\mathbb{E}_{Q}[f(X)]\;\leq \; \alpha
\\ {D\mathstrut}^{\bf t}(QW, \,{\widetilde{Q}\mathstrut}_{0}{\,\widetilde{\!W}\mathstrut}_{\!0}) \,<\, \infty
}}
\!\!\!\!\!\!\!\!\!\!\!\!\!
{F\mathstrut}^{\bf t}(\rho, QW, QW).
\end{equation}
Suppose (\ref{eqRewriteMix1}) is finite and $\hat{Q}\hat{W}$ achieves the minimum on the RHS. Then we can use Lemma~\ref{lem5PPFixedMixed}
with $\hat{Q}\hat{W}$. The rest of the proof is the same as for Theorem~\ref{thmConvergeGradient}.
$\square$

\bigskip

\section{Convergence %of iterative minimization
for fixed $R$ and $\eta$}\label{FixedMixed2}

\bigskip

%Here
In this section
we show convergence of an iterative minimization at a fixed slope $\eta$ in the direction of $\alpha$, i.e., for a given $R$.
Let us define
\begin{align}
{F\mathstrut}^{\bf t}(\eta, QW, \widetilde{Q}\,\widetilde{\!W}, R) \; \triangleq & \;
\max\Big\{{F\mathstrut}_{1}^{\bf t}(QW, \,\widetilde{Q}\,\widetilde{\!W}), \, {F\mathstrut}_{2}(QW, R)\Big\}
\nonumber \\
&
+ \, \eta \,\mathbb{E}_{Q}[f(X)],
\label{eqFMix}
\end{align}
where ${F\mathstrut}_{1}^{\bf t}(QW, \,\widetilde{Q}\,\widetilde{\!W})$ and ${F\mathstrut}_{2}(QW, R)$ are as defined in (\ref{eqF1}) and (\ref{eqF2}), respectively.
\begin{equation} \label{eqE0Mix2}
{E\mathstrut}_{0}^{\bf t}(\eta, \widetilde{Q}\,\widetilde{\!W}, R) \; \triangleq \;
\min_{\substack{\\Q, \, W
}}
{F\mathstrut}^{\bf t}(\eta, QW, \widetilde{Q}\,\widetilde{\!W}, R).
\end{equation}
Here ${E\mathstrut}_{0}^{\bf t}(\eta, \widetilde{Q}\,\widetilde{\!W}, R)$ plays a role of ``${E\mathstrut}_{0}$'' of a supporting line in the variable $\alpha$ of the function
$E(\alpha) = {E\mathstrut}_{\!c}^{\bf t}(\widetilde{Q}\,\widetilde{\!W}, R, \alpha)$, defined in (\ref{eqDef}), as shown by the following lemma.

\bigskip

\begin{lemma} %[]
\label{lemSuppLineMix2}%\newline
{\em For any $\eta \geq 0$ it holds that}
\begin{equation} \label{eqEquivMix2}
{E\mathstrut}_{\!c}^{\bf t}(\widetilde{Q}\,\widetilde{\!W}, R, \alpha) \;\; \geq \;\;
{E\mathstrut}_{0}^{\bf t}(\eta, \widetilde{Q}\,\widetilde{\!W}, R) \, - \, \eta \alpha,
\end{equation}
{\em and there exists %such
$\alpha \geq \min_{\,x} f(x)$
which satisfies (\ref{eqEquivMix2}) with equality.}
%that there is equality in (\ref{eqEquivMix2}).}
\end{lemma}
\begin{proof}
Similar to Lemma~\ref{lemSuppLine}.
\end{proof}

\bigskip

An iterative minimization procedure at a fixed slope $\eta$ is defined as follows.
\begin{align}
&
\begin{array}{l}
\displaystyle
\;\;\;\;\;\;\;\;
{Q\mathstrut}_{\ell}{W\mathstrut}_{\!\ell} \;\;\; \in \;\;
\underset{\substack{\\Q, \, W}}{\arg\min} \;
{F\mathstrut}^{\bf t}(\eta, \,QW, \,{\widetilde{Q}\mathstrut}_{\ell}{\,\widetilde{\!W}\mathstrut}_{\!\ell}, \,R), \\
%\label{eqStep1} \\
{\widetilde{Q}\mathstrut}_{\ell\, + \, 1}{\,\widetilde{\!W}\mathstrut}_{\!\ell\, + \, 1} \;\, = \;\;\; {Q\mathstrut}_{\ell}{W\mathstrut}_{\!\ell},
\end{array}
\label{eqMinProcedureMix2} \\
&
\;\;\;\;\;\;\;\;\;\;\;\;\;\;\;\;\;\;\;\;\;
\;\;\;\;\;\;\;\;\;\;\;\;\;\;\;\;\;\;\;\;\;\;\;\;\;\;\;
\ell \, = \, 0, 1, 2, ... \,.
\nonumber
\end{align}
It is assumed that the set $\big\{QW:  {F\mathstrut}_{1}^{\bf t}(QW, \,{\widetilde{Q}\mathstrut}_{0}{\,\widetilde{\!W}\mathstrut}_{\!0}) < +\infty \big\}$
is non-empty,
which guarantees
${F\mathstrut}^{\bf t}(\eta, \,{Q\mathstrut}_{0}{W\mathstrut}_{\!0}, \,{\widetilde{Q}\mathstrut}_{0}{\,\widetilde{\!W}\mathstrut}_{\!0}, \,R) = {E\mathstrut}_{0}^{\bf t}(\eta, {\widetilde{Q}\mathstrut}_{0}{\,\widetilde{\!W}\mathstrut}_{\!0}, R) < +\infty$.
The iterative procedure results in a monotonically non-increasing sequence ${E\mathstrut}_{0}^{\bf t}(\eta, {\widetilde{Q}\mathstrut}_{\ell}{\,\widetilde{\!W}\mathstrut}_{\!\ell}, R)$, $\ell = 0, 1, 2, ... \,$, as can be seen from (\ref{eqFMix}), (\ref{eqE0Mix2}).
The sequence converges to the global minimum in the set $\big\{\widetilde{Q}\,\widetilde{\!W}: {D\mathstrut}^{\bf t}(\widetilde{Q}\,\widetilde{\!W}, \,{\widetilde{Q}\mathstrut}_{0}{\,\widetilde{\!W}\mathstrut}_{\!0}) \,<\, +\infty\big\}$, as %given by
stated in
the following theorem.

\bigskip

\begin{thm}%[]
\label{thmConvergeMix2}%\newline
{\em Let ${\big\{{Q\mathstrut}_{\ell}{W\mathstrut}_{\!\ell}\big\}\mathstrut}_{\ell \, = \, 0}^{+\infty}$ be a sequence of iterative solutions produced by (\ref{eqMinProcedureMix2}). %\newline
%If $\,{E\mathstrut}_{\!c}({Q\mathstrut}_{0}, R, \alpha) < +\infty$, %\newline
Then}
\begin{equation} \label{eqConvergeMix2}
%F(\rho, \,\eta, \,{U\mathstrut}_{\!\ell}\,{W\mathstrut}_{\!\ell}, \,{Q\mathstrut}_{\ell})
{E\mathstrut}_{0}^{\bf t}(\eta, {\widetilde{Q}\mathstrut}_{\ell}{\,\widetilde{\!W}\mathstrut}_{\!\ell}, R)
\; \overset{\ell \, \rightarrow\,\infty}{\searrow} \;
\min_{\substack{\\\widetilde{Q}, \, \,\widetilde{\!W}:\\ {D\mathstrut}^{\bf t}(\widetilde{Q}\,\widetilde{\!W}\!, \,{\widetilde{Q}\mathstrut}_{0}{\,\widetilde{\!W}\mathstrut}_{\!0})
\,<\, \infty}}
{E\mathstrut}_{0}^{\bf t}(\eta, \widetilde{Q}\,\widetilde{\!W}, R),
\end{equation}
{\em where ${E\mathstrut}_{0}^{\bf t}(\eta, \widetilde{Q}\,\widetilde{\!W}, R)$ is defined in (\ref{eqE0Mix2}) and ${D\mathstrut}^{\bf t}(\cdot\,,\cdot)$ in (\ref{eqCombination}).}
\end{thm}

\bigskip

To prove this theorem, we use a lemma, which is similar to Lemma~\ref{lem5PP}:

\bigskip

\begin{lemma} %[]
\label{lem5PPMix2}%\newline
{\em Let $\hat{Q}\hat{W}$ be such that ${F\mathstrut}_{1}^{\bf t}(\hat{Q}\hat{W}, \,{\widetilde{Q}\mathstrut}_{0}{\,\widetilde{\!W}\mathstrut}_{\!0}) < +\infty$. Then}
%\newline
\begin{align}
& {F\mathstrut}^{\bf t}(\eta, \,{Q\mathstrut}_{0}{W\mathstrut}_{\!0}, \,{\widetilde{Q}\mathstrut}_{0}{\,\widetilde{\!W}\mathstrut}_{\!0}, \, R)
\;\; \leq \;\;
{F\mathstrut}^{\bf t}(\eta, \hat{Q}\hat{W}, \hat{Q}\hat{W}, R)
\nonumber \\
&
\;\;\;\;\;\;\;\;\;\;\;\;\,\,\,
+ \big|{F\mathstrut}_{1}^{\bf t}( \hat{Q}\hat{W} , \, {\widetilde{Q}\mathstrut}_{0}{\,\widetilde{\!W}\mathstrut}_{\!0})  -   {F\mathstrut}_{1}^{\bf t}( \hat{Q}\hat{W} , \, {\widetilde{Q}\mathstrut}_{1}{\,\widetilde{\!W}\mathstrut}_{\!1})\big|^{+}.
\label{eqBoundOne}
\end{align}
\end{lemma}
\begin{proof}
Similar to Lemma~\ref{lem5PP}.
\end{proof}

\bigskip

{\em Proof of Theorem~\ref{thmConvergeMix2}:}
The RHS of (\ref{eqConvergeMix2}) can be rewritten in terms of ${F\mathstrut}^{\bf t}(\eta, QW, \widetilde{Q}\,\widetilde{\!W}, R)$ as:
\begin{equation} \label{eqRewriteMix2}
\min_{\substack{\\\widetilde{Q}, \, \,\widetilde{\!W}:\\ {D\mathstrut}^{\bf t}(\widetilde{Q}\,\widetilde{\!W}\!, \,{\widetilde{Q}\mathstrut}_{0}{\,\widetilde{\!W}\mathstrut}_{\!0})
\,<\, \infty}}
\!\!\!\!\!\!\!\!\!\!\!\!\!\!
{E\mathstrut}_{0}^{\bf t}(\eta, \widetilde{Q}\,\widetilde{\!W}, R)  = \!\!\!\!
\min_{\substack{\\Q, \, W:\\ {D\mathstrut}^{\bf t}(QW, \,{\widetilde{Q}\mathstrut}_{0}{\,\widetilde{\!W}\mathstrut}_{\!0})
\,<\, \infty}}
\!\!\!\!\!\!\!\!\!\!\!\!\!\!
{F\mathstrut}^{\bf t}(\eta, QW, QW, R).
\end{equation}
Suppose (\ref{eqRewriteMix2}) is finite, and let $\hat{Q}\hat{W}$ achieve the minimum on the RHS. Then we can use Lemma~\ref{lem5PPMix2}
with $\hat{Q}\hat{W}$. The rest of the proof is the same as for Theorem~\ref{thmConvergence}.
$\square$

%\newpage

\bibliographystyle{IEEEtran}
%%\bibliography{researchproposal}

\end{document}